\documentclass[12pt]{article}

\usepackage[T1]{fontenc}
\usepackage[latin9]{inputenc}
\usepackage{geometry}
\usepackage{color}
\usepackage{amsmath}
\usepackage{amsthm}
\usepackage{amssymb}
\usepackage{setspace}
\usepackage{esint}
\usepackage{graphicx}
\usepackage{subcaption}
\usepackage{enumerate}
\usepackage[authoryear]{natbib}
\usepackage[unicode=true,bookmarks=true,bookmarksnumbered=false,bookmarksopen=false,
 breaklinks=true,pdfborder={0 0 1},backref=section,colorlinks=true]
 {hyperref}
\hypersetup{linkcolor=blue,urlcolor=blue,citecolor=blue,setpagesize=false,linktocpage}

\makeatletter
\theoremstyle{plain}
\newtheorem{assumption}{\protect\assumptionname}
\theoremstyle{plain}
\newtheorem{prop}{\protect\propositionname}
\newtheorem{lem}{\protect\lemmaname}
\usepackage[margin=28pt,font=small,labelsep=endash]{caption}
\usepackage{amsfonts}
\usepackage{authblk}
\providecommand{\U}[1]{\protect\rule{.1in}{.1in}}
\makeatother

\providecommand{\assumptionname}{Assumption}
\providecommand{\propositionname}{Proposition}

\makeatother

\providecommand{\assumptionname}{Assumption}

\providecommand{\lemmaname}{Lemma}

\begin{document}
\title{Innovation through intra and inter-regional interaction in economic geography\thanks{We are thankful to Steven Bond-Smith, Sofia B. S. D. Castro, Jo\~{a}o Correia da Silva, Pascal Mossay, Pietro Peretto, Anna Rubinchik and Jorge Saraiva for very useful comments and suggestions. We would also like to thank the participants at the Fourth International Workshop  "Market Studies and Spatial Economics", Universit\'{e} Libre de Bruxelles, ECARES,  at the 8th Euro-African Conference on Finance and Economics / Mediterranean Workshop in Economic Theory, Faculty of Economics, University of Porto, and at the 7th Geography of Innovation Conference, University of Manchester. Funding Information:
Japan Society for the Promotion of Science
Grant/Award Number: 
21K04299; 
Funda\c{c}\~{a}o para a Ci\^{e}ncia e Tecnologia
        UIDB/04105/2020, UIDB/00731/2020 and PTDC/EGE-ECO/30080/2017. Part of this research was developed while Jos\'{e} M. Gaspar was a researcher at the Research
		Centre in Management and Economics, Cat\'{o}lica Porto Business School, Universidade Cat\'{o}lica Portuguesa, through the grant CEECIND/02741/2017.}}

\author{Jos\'{e} M. Gaspar\thanks{School of Economics and Management and CEF.UP, University of Porto. Email: jgaspar@fep.up.pt.} \ and Minoru Osawa\thanks{Institute of Economic Research, Kyoto University. Email: osawa.minoru.4z@kyoto-u.ac.jp.}}

\date{\vspace{-5ex}}
\maketitle

\begin{abstract}
We develop a two-region economic geography model with vertical innovations that
improve the quality of manufactured varieties produced in each region. The chance of innovation  depends on the \emph{related variety}, i.e. the importance of interaction between researchers within the same region rather than across different regions. As economic integration increases from a low level, a higher related variety is associated with more agglomerated spatial configurations. However, if the interaction with foreign scientists is relatively more important for innovation, economic activities may (completely) re-disperse after an initial phase of agglomeration due to the increase in the relative importance of a higher chance of innovation in the less industrialized region. This non-monotonic relationship between economic
integration and spatial imbalances may exhibit very diverse qualitative
properties, not yet described in the literature.
\end{abstract}
\bigskip{}



\section{Introduction}

Geographical economics emphasizes the role of endogenous forces in shaping lasting and sizable economic agglomerations in the modern economy. However, in its aim to explain the spatial distribution of economic activities,  there has been a narrow focus on pecuniary externalities through trade linkages.%
\footnote{Reviews
on the literature of geographical economics, or ``new economic geography'' models, are
provided in the monographs by \citet{Fujita-Krugman-Venables-Book1999}, \citet{BaldwinForslidMartinOttavianoRobertNicoud+2003}, and \citet{fujitathisse2013}, or in the papers by \citet{krugman2011new}, \citet{storper2011regions},  \citet{behrens2011tempora} and \citet{gaspar2018prospective}.} 
Other possible sources
of agglomeration economies, such as knowledge externalities and technological
spillovers, are left out \citep{gaspar2018prospective}. 
But in order to understand the processes of (de)-industrialization, it is crucial to develop theories that explore the interaction among multiple spatial linkages. 
We aim to fill this gap by explaining how the intra-regional and inter-regional interaction between researchers impacts knowledge creation and affects the spatial distribution of agents. 
Moreover, we study how the weight of such interactions interplays with economic integration to understand the evolution of the space economy as trade barriers decrease.

We combine the typical pecuniary externalities in geographical economics \citep{Krugman-JPE1991,Fujita-Krugman-Venables-Book1999,BaldwinForslidMartinOttavianoRobertNicoud+2003}
with the spatial diffusion of knowledge spawned from intra-regional
and inter-regional interactions  to infer
about the circular causality between migration and knowledge flows. In the present work, production of knowledge affects the firms' capacity to innovate, which in turn allows the production of higher quality manufactured varieties in a region. The chance of successful innovations depends on the spatial distribution of mobile agents in the economy. Therefore, it is assumed that regional knowledge levels transfer imperfectly between regions, depending on the \textit{related variety} \citep{Frenken-etal-RS2007}, i.e., the relative importance of interaction between agents within
the same region rather than between different regions -- which depends on several factors such as cognitive proximity, cultural factors, diversity of skills and abilities, among others.   We assume further that the increasing complexity of each variety is offset by the available regional quality levels (cf. Section 3.3) generated from knowledge spillovers. Our modeling strategy is such that indirect utility differentials, which govern the migration of mobile agents between regions, are determined solely by trade linkages and by the spatial dimension of regional interaction  (cf. Section 3.4). We thus avoid the explicit use of dynamics for the innovation process.  This allows us for great analytical tractability and to focus on spatial outcomes as a result of pecuniary factors and the \emph{economic geography} of knowledge spillovers \citep{BS2022}.

We show a very diversified set of qualitative predictions regarding the spatial distribution of mobile agents, and, especially, a very rich gallery of possibilities depending on how related variety affects the agglomeration process as economies become more integrated.  
Specifically, if inter-regional interaction is relatively more important for the success of firms' innovation (related variety is low)\footnote{But not too low; cf. Section 4.2.2.}, then an increase in economic integration from a very low level initially fosters agglomeration in a single region. However, above a certain threshold, more integration leads to more symmetric spatial outcomes, because firms find it worthwhile to relocate to the peripheral regions in order to benefit from higher expected profits due to the sizeable pool of agents in the core, which increases the chance of innovation in the deindustrialized region.  Therefore, when related variety is low (but not too low),\footnote{For exceedingly low values of related variety, the symmetric dispersion equilibrium is the unique stable equilibrium in the entire range of economic integration.} our model accounts for a (complete) re-dispersion of economic activities after an initial phase of agglomeration. That is, we are able to uncover a bell-shaped relation between economic integration and spatial development. In this case, knowledge spillovers constitute a local dispersion force that becomes relatively stronger as economic integration brings about the withering of agglomeration forces due to increasing returns to scale in manufacturing. 
But the process of (de)-industrialization is far from trivial; depending on the specific level of related variety, the process of the bell-shaped relationship between economic integration and spatial imbalances occurs with very different qualitative properties, not yet described in the literature.\footnote{To the best of our knowledge.} In fact, we show that increases in related variety are linked to more pronounced agglomerations in the industrialization process, and to more sudden (discontinuous) jumps towards dispersed outcomes, particularly for intermediate levels of economic integration. 

By contrast, when intra-regional interaction is relatively more important, knowledge spillovers become more localized and generate an additional agglomeration force. Re-dispersion becomes altogether impossible because within-region interaction  is too important for innovation to make any deviation to a deindustrialized region worthwhile. 

While the existence of a re-dispersion phase hinges solely on the spatial dispersive or agglomerative nature of technological spillovers, the type of transition between distinct spatial outcomes due to increasing economic integration may differ across various functional forms governing the likelihood of a successful innovation. To better illustrate these points, we draw a brief comparison between the benchmark case, where regional interaction is additive, and a scenario whereby interaction is multiplicative.

The rest of the paper is organized as follows. Section 2 discusses some related literature. Section 3 introduces the spatial economic model and describes its short-run general equilibrium. Section 4 deals with the existence and stability of long-run equilibria. Section 5 studies the relationship between economic integration and spatial outcomes. In Section 6, we provide some comparative statics and a more general form for the firms' success of innovation and a discussion on the robustness of our results.  Finally, Section 7 is left for discussion and concluding remarks.

\section{Literature Review}

The formation of significant economic clusters is intricately linked to agglomeration economies, as the concentration of economic agents in urban centers yields a range of positive effects, both pecuniary and non-pecuniary \citep{Duranton-Puga-HB2004,Duranton-Puga-HB2015}. The spatial economy can thus be seen as the result of trade-offs between such scale economies and the transportation costs incurred by the movement of goods, people, and information \citep{Proost-Thisse-JEL2019}.

However, most theories of endogenous agglmeration do not address how knowledge externalities operate between locations because they deliberately focus on trade linkages as \textit{the} mode of inter-location interaction to investigate the role of pecuniary externalities, as well as to ensure tractability \citep{Fujita-Mori-PRS2005}.  Such a narrow focus enables
researchers to design a microfounded model based on the firms' perspective
using modern tools of economic theory. Nonetheless, it is true that further
development in geographical economics requires modeling the creation and transfer of
knowledge to infer how it affects the
location of economic activities \citep{fujitathisse2013}. In particular, the role of K-linkages has become increasingly
relevant in the economic geography literature. Building upon pioneering works such as
\citet{Berliant-Fujita-IER2008,Berliant-Fujita-IJET2009,Berliant-Fujita-RSUE2012},
one should hope that a new comprehensive economic geography theory fully integrates
the linkage effects among consumers and producers and K-linkages in
space. According to \citet{Fujita-RSUE2007}, geography is an essential
feature of knowledge creation and diffusion. For instance, people
residing in the same region interact more frequently and thus contribute
to develop the same, regional set of cultural ideas. However, while
each region tends to develop its unique culture, the economy as a
whole evolves according to the synergy that results from the interaction
across different regions (i.e., different cultures). That is, according
to \citet{Duranton-Puga-AER2001}, knowledge creation and location
are inter-dependent. \citet{Berliant-Fujita-RSUE2012} developed a
model of spatial knowledge interactions and showed that higher cultural
diversity, albeit hindering communication, promotes the productivity
of knowledge creation. This corroborates the empirical findings of
\citet{Ottaviano-Peri-JEG2006,Ottaviano-Peri-DP2008}. \citet{Ottaviano-Prarolo-JRS2009}
show how improvements in the communication between different cultures
fosters the creation of multicultural cities in which cultural diversity
promotes productivity. This happens because better communication allows
different communities to interact and benefit from productive externalities
without risking losing their cultural identities. \citet{Berliant-Fujita-SEJ2011}
take a first step towards using a micro-founded R\&D structure to
infer about its effects on economic growth. They find that long-run
growth is positively related to the effectiveness of interaction among
workers as well as the effectiveness in the transmission of public
knowledge.

Therefore, combining the typical pecuniary externalities in geographical economics
with the spatial diffusion of knowledge spawned from intra-regional
and inter-regional interactions alike is important if we want to infer
about an eventual circular causality between migration and the circulation
of knowledge. In other words, geographical economics may shed light on the importance
of knowledge exchanged between different regions through trade networks
compared to ``internally'' generated knowledge.

Reinforcing the importance of heterogeneity in knowledge, it is  crucial
to discern about the \emph{relatedness} of variety. This relatedness
measures the cognitive proximity and distance between sectors that
allows for a higher intensity of knowledge spillovers. According to
\citet{Frenken-etal-RS2007}, a higher \emph{related variety} increases
the inter-sectoral knowledge spillovers between sectors that are technologically
related. This potentially adds a new dimension to the role of heterogeneity
and location in the creation and diffusion of knowledge. \citet{Tavassoli-Carbonara-SBE2014}
have tested the role of knowledge intensity and variety using regional
data for Sweden and found evidence that different types of cognitive
proximity have an important weight. This confirms the relevance of
the spatial determinants of innovation and knowledge creation.

In this paper, knowledge creation and diffusion amounts to vertical innovations in the manufacturing sector, which in turn draws from the literature on endogenous growth. Particularly in Schumpeterian growth theory, innovations
that affect the quality of produced goods or a firm's cost efficiency are usually driven
by stochastic processes, because the production of knowledge involves
some sort of uncertainty. Therefore, we may think of quality as a
proxy  for a given firm's stock of knowledge. 
In  the growth models developed e.g. by \citet{Aghion-Howitt-ECTA1990,Aghion-Howitt-Book1998},
\citet{Young-JPE1998}, \citet{Peretto-JEG1998}, \citet{Howitt-JPE1999},
or more recently \citet{Dinopoulous-Segerstrom-JDE2010}, innovations
occur with a probability that depends on factors such as the amount
of the firm's research effort, the common pool of public knowledge
available to all firms, and the individual firm's quality level. Introducing geography
and worker mobility in these frameworks allows the innovation success to also depend on the magnitude of regional interaction through the
exchange of ideas between workers and producers alike among regions.
If cognitive proximity is more important for innovation, then related variety is high and knowledge spillovers become more localized, thus promoting agglomeration \citep{baldwin2004agglomeration}. If interaction with foreign agents is more important, related variety is low and is  likely to induce the dispersion of economic activities. 

We follow the literature of \emph{Quantitative Spatial Economics} \citep{redding2017quantitative,BEHRENS2021103348,kleinman2023linear} and purposefully abstract from modeling growth or the explicit use of dynamics in the innovation sector, thus ensuring analytical tractability. However, knowledge transfers \emph{imperfectly} between regions, in accordance with the findings in the seminal work of \citet{audretsch1996r}. Nonetheless, our modeling of the innovation sector renders the model scale-neutral. As such, its spatial outcomes are driven by the spatial nature of knowledge spillovers (and by its interplay with pecuniary externalities), rather than by implicit assumptions about scale returns in the innovation sector.\footnote{For a comprehensive discussion on how these  implicit assumptions regarding returns to scale generate mistaken conclusions in geographical economics, we refer the  reader to \cite{BS2022}.}

\section{The model}

The following is an analytically solvable footloose entrepreneur model \citep{BaldwinForslidMartinOttavianoRobertNicoud+2003}. 
The economy is comprised of two regions
indexed by $i=\{1,2\}$, two kinds of labour, two productive sectors
and one R\&D sector. There is a unit mass of (skilled) inter-regionally
mobile agents (which we dub scientists henceforth) and a mass $l\equiv\lambda>0$
of (unskilled) immobile workers (just workers, for short) which are
assumed to be evenly distributed across both regions, i.e., $l_{i}=\frac{\lambda}{2}$
$(i=1,2)$. The amount of scientists in region $1$ is given by $z_{1}\equiv z\in\left[0,1\right]$
and fully describes the spatial distribution of agents in the economy.

\subsection{Demand}

The utility function of a consumer located in region $i$ is given
by: 
\begin{equation}
u_{i}=\mu\ln\left(\frac{M_{i}}{\mu}\right)+B_{i},\ \ \mu>0\label{eq:utility}
\end{equation}
where $B_{i}$ is the num\'{e}raire good produced under perfect competition
and constant returns to scale. This good is produced one-for-one using
$L$ workers and its price is set to unity as is the wage paid to
workers. The quality-augmented CES composite $M_{i}$ is given by: 
\begin{equation}
M_{i}=\left[\sum_{j=1}^{2}\int_{s\in S}\left(\delta_{ij}^{m(s)}d_{ij}(m,s)\right)^{\tfrac{\sigma-1}{\sigma}}ds\right]^{\tfrac{\sigma}{\sigma-1}},\label{eq:ces utility}
\end{equation}
where $d_{ij}(m,s)$ is the demand for manufactures in region $i$ produced
in region $j$ for a given variety $s\in S$ of quality $m \in\left\{ 1,....,k\right\}$, $k$ being the highest quality rung achieved for any variety, $S$ is the mass of
varieties in region $i$ and $\sigma$ is the elasticity of substitution
between any two varieties. The parameter $\delta>1$ indexes the
step size of quality improvements in region $i$ after a successful
innovation and $k$ is the leading quality grade for any given variety
$s$. Since $\delta^{m(s)}$ is increasing in $m$, the utility in (\ref{eq:ces utility})
reflects the fact that consumers have a preference for higher quality
(\citealp{Dinopoulos-Segerstrom-DP2006}). However, love for variety implies
that varieties, once adjusted for quality, are perfect substitutes.
This means that each consumer purchases only the good with the lowest
quality adjusted price, $p_{i}(s)/\delta_{i}^{m}$. If any two goods
have the same quality adjusted price, consumers will only buy the
highest quality good (Dinopoulos and Segerstorm \citeyear{Dinopoulos-Segerstrom-DP2006,Dinopoulous-Segerstrom-JDE2010};
\citealp{davis2012private}). The firm responsible for each quality improvement for a variety $s$ retains a monopoly right to produce that variety at the highest quality possible, i.e., $k$. Therefore, if the quality rungs $m=1,...,k$ have been reached, the $k\textit{th}$ innovator is the sole source of the good of variety $s$ with the quality level $\delta^k$ \citep{Barro-Sala-i-Martin-Book2004}. In what follows, we can already predict that the short-run general equilibrium will be comprised solely of firms that produce the highest quality possible of their variety  $s \equiv s_k\in S$.

Since individual incomes depend on the distribution of labour activities,
we have $y_{i}=1$ for the workers, and $y_{i}=w_{i}$, which is the
the compensation paid to the scientists that engage in research. Therefore,
the regional income is given by: 
\[
Y_{i}=\frac{\lambda}{2}+w_{i}z_{i}.
\]

\noindent The individual budget constraint is given by $B_{i}+P_{i}M_{i}=y_{i}+\bar{B}$, where $\bar{B}$ is the initial
endowment of the num\'{e}raire and $P_i$ is the regional price index. Since total expenditure on manufactures must equal $P_{i}$
times the quantity of composite $M_{i}$, agents maximize (\ref{eq:utility}) subject to the following budget constraint:
\[
B_{i}+\sum_{j=1}^{2}\int_{s\in S}p_{ij}(s)d_{ij}(s)ds=y_{i}+\bar{B},
\]
which yields the following optimal individual demands: 
\begin{equation}
d_{ij}(s)=\mu\frac{a_{i}(s)p_{ij}(s)^{-\sigma}}{P_{i}^{1-\sigma}},\ B_{i}=y_{i}+\bar{B}-\mu,\ M_{i}=\mu P_{i}^{-1},\label{eq:optimal individual demand}
\end{equation}
where  $a_{i}(s)=\delta_{i}^{m(s)(\sigma-1)}$ is just an alternative
measure of a variety $s$'s quality in region $i$ and $P_{i}$ is the quality adjusted price index
given by: 
\begin{equation}
P_{i}=\left[\sum_{j=1}^{2}\int_{s\in S}a_{ij}(s)p_{ij}(s)^{1-\sigma}ds\right]^{\tfrac{1}{1-\sigma}}.\label{eq:quality price index}
\end{equation}

\noindent We assume that $\bar{B}>\mu$ in order to assure that both
types of goods are consumed. From (\ref{eq:utility}) and (\ref{eq:optimal individual demand}),
we obtain the indirect utility:
\begin{equation}
v_{i}=y_{i}-\mu\ln P_{i}-\mu+\bar{B}.\label{eq:indirect utility}
\end{equation}

\subsection{Manufacturing firms}

For each firm, there is a variable input requirement of $\beta$ workers.
A manufacturing firm in region $i$ thus faces the following cost:
\begin{equation}
C_{i}(q_{i}(s))=\beta q_{i}(s),\label{eq:total cost function}
\end{equation}
where $q_{i}$ is total production by a firm in region $i$.

Trade of manufactures between regions is burdened by transportation
costs of the iceberg type. Let the iceberg costs $\tau_{ij}\in(1,+\infty)$
denote the number of units that must be shipped at region $i$ for
each unit that is delivered at region $j$. We have $\tau_{ij}=\tau_{j}\in\left(1,\infty\right)$
for $i\neq j$ and $\tau_{ij}=1$ otherwise. The quantity produced
by a firm in region $i$ is thus given by: 
\[
q_{i}(s)=\sum_{j=1}^{2}\tau_{ij}d_{ij}\left(1+z_{j}\right).
\]

\noindent The profit of a manufacturing firm with the leading grade
$k$ of variety $s$ in region $i$ is given by: 
\begin{align}
\tilde{\pi}_{i}\equiv\pi_{i}(s\equiv s_{k})= & \sum_{j=1}^{2}p_{ij}(s_{k})d_{ij}(s_{k})\left(\frac{\lambda}{2}+z_{j}\right)-\beta q_{i}(s_{k})\nonumber \\
= & \sum_{j=1}^{2}\left(p_{ij}(s_{k})-\tau_{ij}\beta\right)d_{ij}(s_{k})\left(\frac{\lambda}{2}+z_{j}\right),\label{eq:profit of firm i}
\end{align}
where $s_k$ denotes the highest quality of a given variety $s$. Given (\ref{eq:profit of firm i}) and the optimal individual demand
in (\ref{eq:optimal individual demand}), the firm's profit maximizing
price is the usual mark-up over marginal cost: 
\begin{equation}
p_{ij}=\dfrac{\sigma}{\sigma-1}\tau_{ij}\beta,\label{eq:optimal price}
\end{equation}
which does not depend on the quality of the firm's variety $s$. Under
(\ref{eq:quality price index}), the regional quality adjusted price
index in (\ref{eq:quality price index}) becomes: 
\begin{equation}
P_{i}=\dfrac{\beta\sigma}{\sigma-1}\left(\sum_{j=1}^{2}\phi_{ij}A_{j}n_{j}\right)^{\tfrac{1}{1-\sigma}},\label{eq:quality price index 2}
\end{equation}
where $\phi_{ij}\equiv\tau_{ij}^{1-\sigma}\in(0,1)$ is the \emph{freeness
of trade} and $n_{i}$ is the number of manufacturing firms (and manufactured
varieties), and: 
\[
A_{i}=\sum_{s=1}^{n_{i}}a_{i}(s),
\]
is the aggregate quality index for region $i$. The latter also constitutes
a measure of the aggregate regional knowledge level. 

Consider now that the aggregate knowledge level in one region is a
public good such that it becomes readily available to the other region.
Then $A_{i}=\max\left\{ A_{1},A_{2}\right\} $ so that both regions
are able to reach the highest aggregate knowledge level available.
This leads to the following assumption.
\begin{assumption}
$A_{1}=A_{2}=A$.
\end{assumption}
The rationale behind Assumption 1 is basically to eliminate first-nature advantages between regions (i.e., regional asymmetries). This not only makes the analysis much simpler but also allows us
to focus on conveying the main message of the paper, which is the
effects of regional interaction on the spatial distribution of economic
activities. Under Assumption 1, the regional price index becomes simply:
\begin{equation}
P_{i}=\dfrac{\beta\sigma}{\sigma-1}\left(A\sum_{j=1}^{2}\phi_{ij}n_{j}\right)^{\tfrac{1}{1-\sigma}}.\label{eq:price index 10}
\end{equation}
The higher the aggregate quality, the lower the cost of living in
region $i$.

\subsection{R\&D sector}

We assume that there is free entry in the R\&D sector for each variety
$s$. In order to innovate, a firm decides \emph{ex-ante} to employ
$\alpha$ scientists in the R\&D sector. In doing so, a firm producing
variety $s$ at region $i$ reaches a leading quality grade $k$ with
instantaneous probability: 
\begin{equation}
\Phi_{i}(s)=\min\left\{ \dfrac{bz_{i}+(1-b)z_{j}}{a_{i}(s)}\gamma A,1\right\} ,\label{eq:prob of succesful innovation-1}
\end{equation}
where $b\in(0,1)$ is the \emph{weight of intra-regional interaction
in the chance of innovation }and defines the importance of the exchange
of ideas between researchers alike among regions, and $\gamma>0$
is an efficiency parameter. We are assuming that the lowest quality
grade possible is $\underline{a}>1$ so that $a_{i}(s)>1$, for any
$s$. It is also reasonable to assume that the firm's research success is
greater the higher the level of aggregate knowledge in region $i$, $A_{i}=A$,
available to all firms alike. Finally, we assume that the innovation
rate is decreasing in the complexity of each product, as measured
by its quality level $a_{i}(s)$ \citep{Li-AER2003,Dinopoulous-Segerstrom-JDE2010}.

It is worthwhile explaining how the underlying specification governing
the innovation process depends on the magnitude of the interactions
(or lack of them) between different ``sets of ideas''. Analogously
to the interpretation of \citet{Berliant-Fujita-RSUE2012}, we implicitly
assume that each region holds its own set of ideas (or culture). Therefore,
production of knowledge (which amounts to innovation), depends on
the amount of ``within region'' interaction among scientists, but
also on the interaction with scientists hailing from a different region.
As such, we can say that $b$ is sort of a measure of the \emph{related
variety} of knowledge \citep{Frenken-etal-RS2007}. A higher related
variety means that innovation benefits more from a regional common
pool of ideas, that is, from the interaction between researchers and
workers in the same region. Specifically, when $b\in(1/2,1)$ -- related variety is \emph{high} --, we have that $\Phi_i(s)$ is increasing in $z$ so that innovation in region $1$ is more likely the more researchers live there (technological spillovers are more localized). Conversely, when $b\in(0,\frac{1}{2})$, we say that related variety is \emph{low}, and innovation benefits from more researchers living in the other region. In this case, congestion dampens innovation and generates an additional dispersion force in the model.  

We may look at the term $\gamma A\left[bz_{i}+(1-b)z_{j}\right]$ as the spatially-weighted average of global knowledge observed by each firm. We thus introduce a spatial mechanism of localized spillovers such that knowledge transfers imperfectly between researchers in different locations \citep{BS2022}.   In Section 6.2, we briefly discuss a more general case for the success of innovation that preserves this spatial mechanism of technological spillovers.  Nonetheless, we advance the argument that the additive (linear) specification in (\ref{eq:prob of succesful innovation-1}) is suitable as a benchmark case, because it provides detailed and interesting insights, without sacrificing analytical tractability.



When a manufacturing firm in region $i$, that produces variety $s$,
innovates, it has the probability $\Phi_{i}(s)$ of reaching the leading
quality grade $k$ of its variety. A firm that successfully reaches the grade quality
$k$ becomes the quality leader and charges the monopolistic competitive
price, collecting profits $\tilde{\pi}_{i}\equiv\pi_{i}(s_{k}).$
Lower quality products are considered obsolete. Firms who are unable
to attain the leading quality grade face \emph{creative destruction}
and are priced out of the market \citep{Aghion-Howitt-ECTA1990}.

\subsection{Short-run equilibrium}

Given the research intensity $\alpha$, each firm faces the following
expected profit:
\[
E\left[\pi_{i}(s)\right]=\Phi_{i}(s)\tilde{\pi}_{i}-\left[1-\Phi_{i}\left(s\right)\right]0-\alpha w_{i},
\]
where $w_{i}$ is the nominal wage paid to scientists in region $i.$ 

Labour market clearing implies that the number of varieties (and hence
firms) is given by $S=z_{i}/\alpha$, from where the price index in
(\ref{eq:price index 10}) becomes:

\begin{equation}
P_{i}=\dfrac{\beta\sigma}{\sigma-1}\left(\frac{A}{\alpha}\sum_{j=1}^{2}\phi_{ij}z_{j}\right)^{\tfrac{1}{1-\sigma}}.\label{eq:price index final}
\end{equation}

\noindent Given free-entry in the R\&D sector, in equilibrium expected
profits are driven down to zero, which yields the following condition:
\[
w_{i}=\frac{\Phi_{i}(s)\tilde{\pi}_{i}}{\alpha}.
\]
Using (\ref{eq:optimal individual demand}), (\ref{eq:optimal price}),
this becomes:
\begin{equation}
w_{i}=\dfrac{\mu a_{i}}{\alpha\sigma}\Phi_{i}\sum_{j=1}^{2}\left(\dfrac{p_{ij}}{P_{i}}\right)^{1-\sigma}\left(\frac{\lambda}{2}+z_{j}\right).\label{eq:first wage}
\end{equation}
Replacing (\ref{eq:optimal price}), (\ref{eq:prob of succesful innovation-1}) and (\ref{eq:price index final})
in (\ref{eq:first wage}) yields:
\begin{equation}
w_{i}=\frac{\mu\gamma}{\sigma}\left[bz_{i}+(1-b)z_{j}\right]\left(\dfrac{\frac{\lambda}{2}+z_{i}}{z_{i}+\phi z_{j}}+\phi\dfrac{\frac{\lambda}{2}+z_{j}}{\phi z_{i}+z_{j}}\right).\label{eq:wage equation}
\end{equation}

\noindent Finally, using (\ref{eq:wage equation}) and (\ref{eq:indirect utility}),
we get the indirect utility in region $i$:
\begin{equation}
v_{i}=\frac{\mu\gamma}{\sigma}\left[bz_{i}+(1-b)z_{j}\right]\left(\dfrac{\frac{\lambda}{2}+z_{i}}{z_{i}+\phi z_{j}}+\phi\dfrac{\frac{\lambda}{2}+z_{j}}{\phi z_{i}+z_{j}}\right)+\dfrac{\mu}{\sigma-1}\ln\left[z_{i}+\phi z_{j}\right]+\eta,\label{eq:indirect utility final}
\end{equation}
where $\eta\equiv-\mu\left(\frac{\beta\sigma}{\sigma-1}\right)+\frac{\mu}{\sigma-1}\ln\left(\frac{A}{\alpha}\right)-\mu+\bar{B}$
is a constant. 

Notice how our assumptions on the innovation rate in (\ref{eq:prob of succesful innovation-1}) imply that regional knowledge, $a(i)$ and $A$, does not affect indirect utility. This means that we can conveniently avoid modelling the dynamics of innovation. As a result, the spatial outcomes of our model are solely determined by pecuniary factors and by the spatial dimension of knowledge creation and diffusion, captured by the term $bz_i+(1-b)z_j$, as proposed by \citet{BS2022}.

\section{Long-run equilibria}

Scientists are free to migrate between regions. In doing so, they
choose the region that offers them the highest indirect utility. The
long-run spatial distribution thus depends on the utility differential:
\begin{equation}
\Delta v(z)=v_{1}(z)-v_{2}(z).\label{eq:utility differential}
\end{equation}
 We follow \cite{Castro_2021} in the characterization of equilibria and their stability. There are two kinds of long-run equilibria which should be dealt
with separately. 
\begin{enumerate}
\item \emph{Agglomeration }of all scientists in a single region $z^{*}=\left\{ 0,1\right\} $
is an equilibrium if and only if $\Delta v(1)\geq0,$ or, equivalently,
$\Delta v(0)\leq0$.
\item \emph{Dispersion} of scientists $z^{*}\in\left(0,1\right)$ is an
equilibrium if and only if $\Delta v(z^{*})=0$. If $z^{*}=\frac{1}{2}$
it corresponds to \emph{symmetric dispersion.} Otherwise, it is called \emph{asymmetric}.
\end{enumerate}
Equilibria are stable if, after a perturbation such that $z=z^{*}\pm\epsilon$,
with $\epsilon>0$ small enough, the utility differential $\Delta v(z)$
becomes such that agents go back to their place of origin, i.e., $z=z^{*}$.
A sufficient condition for stability of agglomeration is $\Delta v(1)>0$
(or $\Delta v(0)<0$). A sufficient condition for the stability of dispersion
is that $\Delta v^{\prime}(z^{*})<0$. When equilibria are \emph{regular} (resp. $\Delta v(1)\neq0$ and $\Delta v^{\prime}(z^{*})\neq0$), these conditions are also necessary.\footnote{In models that are well-behaved, the non-existence of
irregular long-run equilibria holds in a full measure subset of a suitably defined parameter
space.} 

\subsection{Existence and multiplicity}

Our first result regards the multiplicity of long-run equilibria. Given symmetry across regions, we focus on the case whereby region $1$ is either the same size or is larger than region $2$, i.e., $z\in\left[\frac{1}{2},1\right]$. 

Symmetric dispersion $z^*=\frac{1}{2}$ is called an \emph{invariant pattern}, because it is a long-run equilibrium for the entire parameter range \citep{aizawa2020break5}. Next, we have the following result regarding possible equilibria for $z\in\left(\frac{1}{2},1\right].$

\begin{prop}
There are at most two equilibria for $z\in\left(\frac{1}{2},1\right].$
\end{prop}
\begin{proof}
See Appendix A.
\end{proof}

We can be more precise regarding the existence of dispersion equilibria with the following result.

\begin{prop}
A dispersion equilibrium $z\equiv z^{*}\in\left(\frac{1}{2},1\right]$
exists if $b\in\left(\max\left\{ 0,\tilde{b}\right\} ,\hat{b}\right)$,
where:{\small{}
\[
\tilde{b}\equiv\frac{\gamma(\sigma-1)(2z-1)\left[(z-1)z\left(\phi^{2}-1\right)+\phi^{2}\right]+\sigma\left[z(\phi-1)+1\right]\left[z(\phi-1)-\phi\right]\ln\left[\frac{z(\phi-1)+1}{z(1-\phi)+\phi)}\right]}{\gamma(\sigma-1)(2z-1)(\phi+1)\left[2(z-1)z(\phi-1)+\phi\right]},
\]
}and 
\[
\hat{b}\equiv\frac{1+\phi^{2}}{(1+\phi)^{2}}.
\]
\end{prop}

\begin{proof}
See Appendix A.
\end{proof}

\begin{figure}
     \centering
     \begin{subfigure}[b]{0.47\textwidth}
         \centering
         \includegraphics[width=\textwidth]{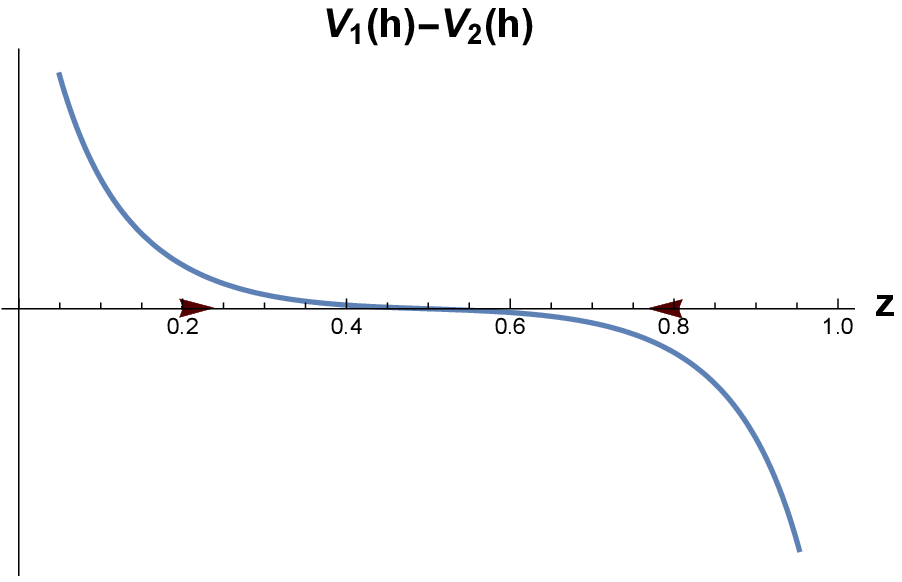}
         \caption{Stable symmetric dispersion: $\phi=0.1$.}
         \label{fig:scenariosa}
     \end{subfigure}
     \hfill
     \begin{subfigure}[b]{0.47\textwidth}
         \centering
         \includegraphics[width=\textwidth]{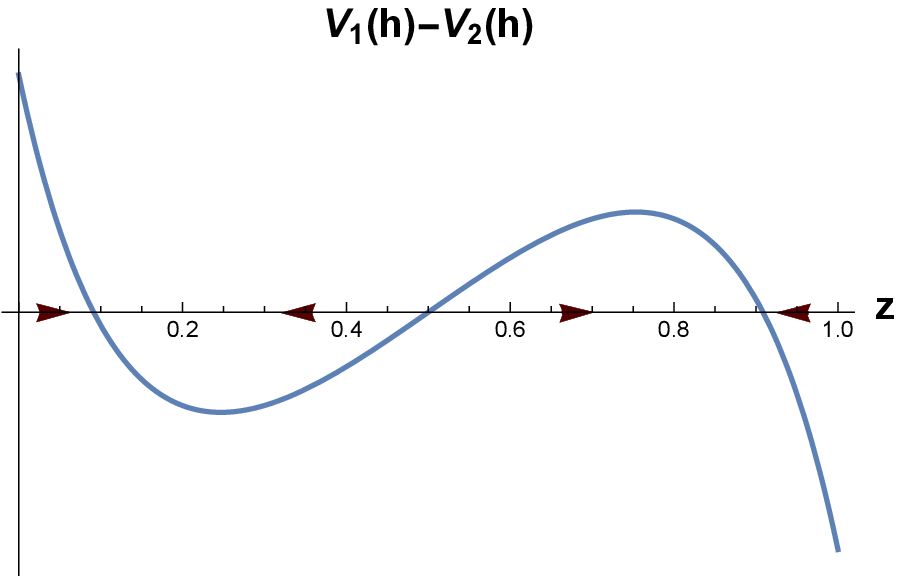}
         \caption{Stable asymmetric dispersion: $\phi=0.3$.}
         \label{fig:scenariosb}        
     \end{subfigure}
     \par\bigskip
          \begin{subfigure}[b]{0.47\textwidth}
         \centering
         \includegraphics[width=\textwidth]{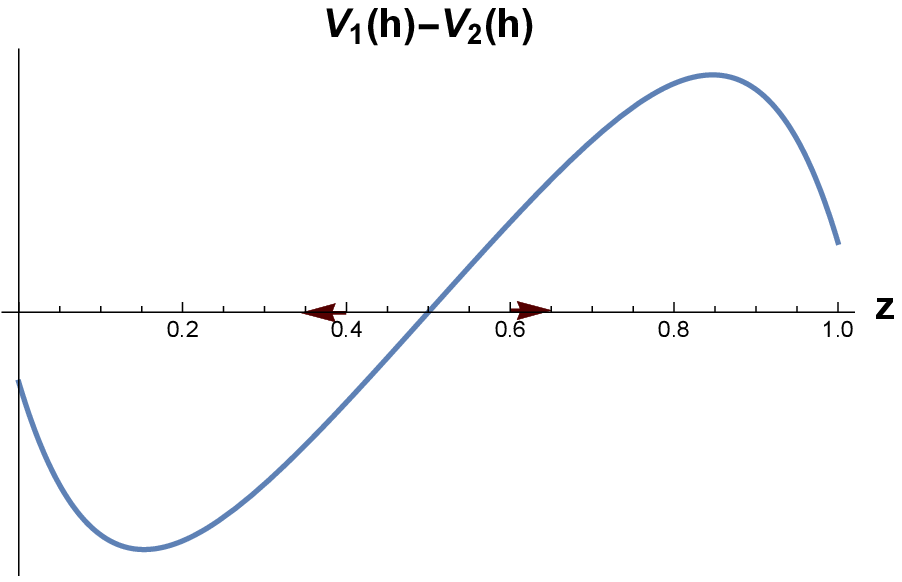}
         \caption{Stable agglomeration: $\phi=0.38$.}
         \label{fig:scenariosd}
     \end{subfigure}
     \hfill
\begin{subfigure}[b]{0.47\textwidth}
         \centering
         \includegraphics[width=\textwidth]{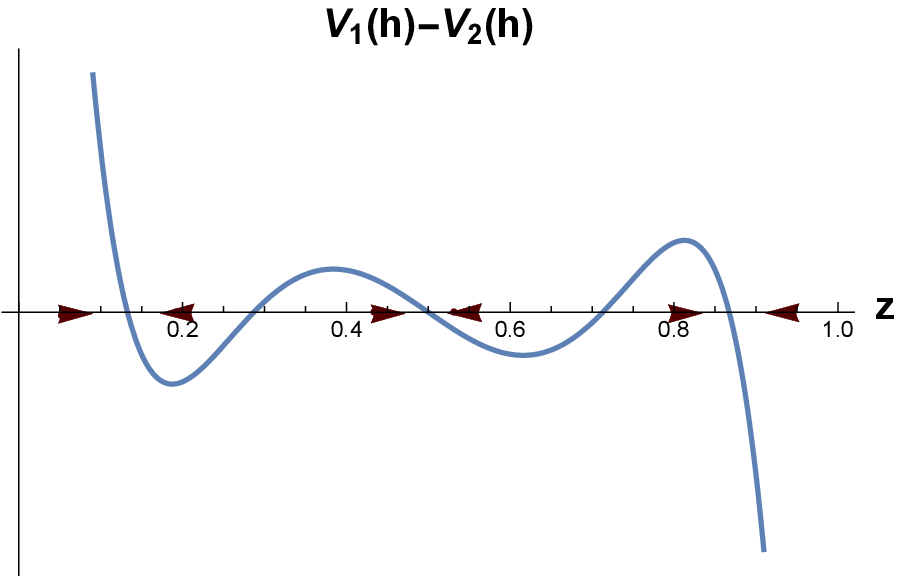}
         \caption{Stable asymmetric dispersion: $\phi=0.4$.}
         \label{fig:scenariosc}        
     \end{subfigure}
     \par\bigskip
     \begin{subfigure}[b]{0.47\textwidth}
         \centering
         \includegraphics[width=\textwidth]{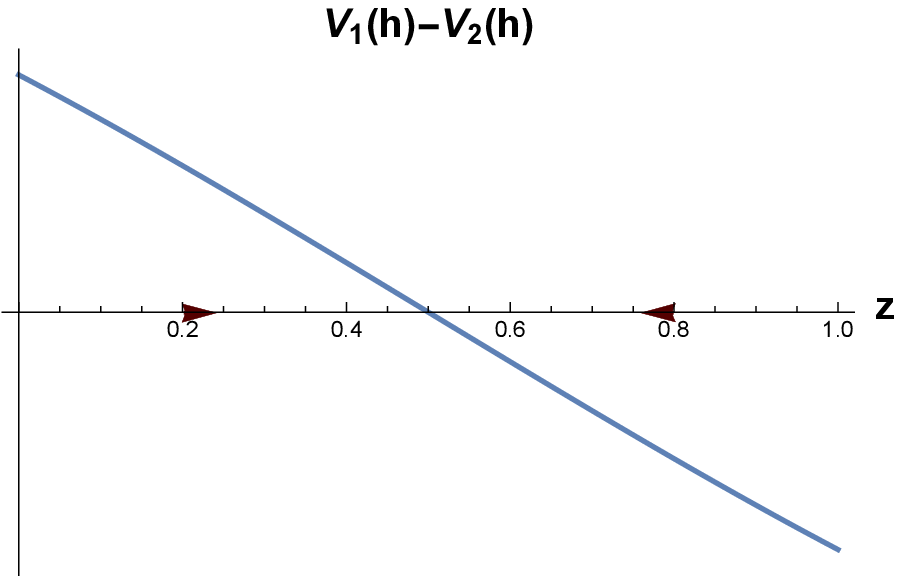}
         \caption{Stable symmetric dispersion: $\phi=0.8$.}
         \label{fig:scenariose}        
     \end{subfigure}     
\caption{Long-run equilibria and their stability as $\phi$ increases.}
\label{fig:Figure1}
\end{figure}

Figure~\ref{fig:Figure1} shows one scenario with five qualitatively different cases, with varying freeness of trade $\phi$, for  $b\in (0,\frac{1}{2})$, regarding existence of the model's long-run spatial distribution, which exhaust all mathematical possibilities for the parameter values $\left(\lambda,\gamma,\sigma,b\right)=\left(2,1,5,0.342\right)$. As a prelude to the forthcoming Section, Figure~\ref{fig:Figure1} also numerically depicts the local stability of each equilibrium, which is to be analysed analytically in greater detail in Section 3.2.

In Figure~\ref{fig:scenariosa}, only symmetric dispersion exists and is stable for a very small $\phi$. For a higher trade freeness we have one stable asymmetric dispersion for $z\in\left(\frac{1}{2},1\right]$ as portrayed in Figure~\ref{fig:scenariosb}. For a  greater $\phi$, Figure~\ref{fig:scenariosd} shows that the asymmetric dispersion equilibrium disappears and symmetric dispersion becomes unstable, whereas agglomeration becomes stable. For an even greater $\phi$, Figure~\ref{fig:scenariosc} illustrates an example of two long-run dispersion equilibria $z^*$ for $z\in\left(\frac{1}{2},1\right]$ and  symmetric dispersion $z=\frac{1}{2}$, whereby  we can observe that both symmetric dispersion and the more agglomerated dispersion equilibrium are locally stable, whereas the less agglomerated equilibrium is unstable. The economy re-disperses and agglomeration does not exist in this particular case.  However, as the trade freeness increases further,   symmetric dispersion remains stable and all other equilibria disappear. 

When related variety is low, $b\in (0,\frac{1}{2})$, the model accounts for a ``bubble-shaped'' relationship between economic integration and spatial imbalances  \citep{Pfluger-Suedekum-JUE2008}, whereby firms are initially dispersed, then start to agglomerate in a single region as the trade freeness increases, but then find it worthwhile to relocate to the peripheral regions in order to benefit from higher expected profits due to the sizeable pool of scientists in the core which increases the chance of innovation in the periphery. In other words, when related variety is low, knowledge spillovers constitute a local dispersion force whose strength becomes relatively higher as economic integration increases and leads to the vanishing of agglomeration forces.

The case of high related variety, $b\in (\frac{1}{2},1)$, is much less diversified and can be accounted for resorting to a subset of the pictures from Figure \ref{fig:Figure1}. The history as economic integration increases is as follows. For a very low trade freeness, symmetric dispersion is the only stable equilibrium as in Figure~\ref{fig:scenariosa}. For an intermediate value of $\phi$, one asymmetric dispersion equilibrium arises which is the only stable one and becomes more asymmetric as $\phi$ increases further. This is akin to the picture in Figure~\ref{fig:scenariosb}. Finally, the asymmetric dispersion equilibrium gives rise to stable full agglomeration in one single region once $\phi$ becomes very high. This is illustrated in Figure~\ref{fig:scenariosd}. In other words, when intra-regional interaction is relatively more important, knowledge spillovers are more localized, and thus constitute an additional agglomeration force. 

In the forthcoming Sections, we will analytically and numerically study in greater detail   the local stability of the spatial distributions and the qualitative change in the model's structure as economic integration increases.

\subsection{Stability}

\subsubsection{Agglomeration}

Regarding agglomeration, using (\ref{eq:indirect utility final})
in (\ref{eq:utility differential}), we have that it is stable if:
\[
\mathcal{S} \equiv \frac{\gamma\left[(b-1)(\lambda+2)\phi^{2}+2b(\lambda+1)\phi+(b-1)\lambda\right]}{2\sigma\phi}-\frac{\ln\phi}{\sigma-1}>0.
\]
The second term is positive. Hence, agglomeration is always stable
if the first term is also positive:
\[
b>b_s \equiv \frac{(\lambda+2)\phi^{2}+\lambda}{(\phi+1)\left[(\lambda+2)\phi+\lambda\right]}.
\]
It is easy to check that $b_{s}<\tfrac{1}{2}$ if $\phi\in\left(\frac{\lambda }{\lambda +2},1\right)$,
which means that, if $\phi\in\left(\frac{\lambda }{\lambda +2},1\right)$ and $b>\frac{1}{2}$,
agglomeration is stable. In any case, we can conclude that agglomeration
is always stable if related variety is extremely high. 

Let us now define as
\emph{sustain point} \citep{Fujita-Krugman-Venables-Book1999}, a value of $\phi$ such that $\mathcal{S}(\phi)=0.$
We have the following result relating the freeness of trade and the
relatedness of variety.

\begin{prop}
If $b<\tfrac{1}{2}$, there exist at most two sustain points, $\phi_{s1}$
and $\phi_{s2},$ and agglomeration is unstable for $\phi\in \left(0,\phi_{1s}\right)\cup\left(\phi_{2s},1\right) $
and stable for $\phi\in(\phi_{1s},\phi_{2s})$. If $b>\frac{1}{2}$,
there exists a unique sustain point $\phi_{s1}$ and agglomeration
is unstable for $\phi\in(0,\phi_{1s})$ and stable if $\phi\in\left(\phi_{1s},1\right)$.
\end{prop}\label{prop3}

\begin{proof}
See Appendix A.
\end{proof}

The result in Proposition 3 suggests that an intermediate level of economic integration favours agglomeration only if the
interaction with foreign scientists is relatively more important for
the chance of successful innovation, i.e. if related variety is not too high. By contrast, if the within region
interaction of scientists is more important, agglomeration is only
possible when the freeness of trade is high enough.

\subsubsection{Symmetric dispersion}

Regarding symmetric dispersion $z^{*}=\frac{1}{2}$, using (\ref{eq:indirect utility final})
in (\ref{eq:utility differential}) we can say that it is stable if:
\begin{equation}
\mathcal{B}\equiv \gamma(\sigma-1)\left[2b(\lambda+1)(\phi+1)^{2}-(2\lambda+3)\phi^{2}-2\lambda-1\right]+2\sigma\left(1-\phi^{2}\right)<0.\label{eq:stability symmetric dispersion}
\end{equation}
 In fact, it is always unstable if the first term is positive, i.e.
if:
\[
b>\bar{b}\equiv\frac{(2 \lambda +3) \phi ^2+2 \lambda +1}{2 (\lambda +1) (\phi +1)^2}.
\]

\noindent This means that if related variety is prohibitively high,
symmetric dispersion is surely unstable.

We can observe that $\mathcal{B}$ in (\ref{eq:stability symmetric dispersion}) is a second
degree polynomial in $\phi$ with at most two zeros, i.e.,\emph{ break
points }$\phi_{b1}$ and $\phi_{b2}$, with $\phi_{b1}<\phi_{b2}$
and has a negative leading coefficient. Therefore if both break points
exist, we have that symmetric dispersion is stable for $\phi\in(0,\phi_{b1})\cup (\phi_{b2},1)$ and
unstable for $\phi\in\left(\phi_{b1},\phi_{b2}\right)$. The expressions for the break points, along with the conditions for their existence, are provided in Appendix A.4. We have the following result.

\begin{prop}
For $b>\frac{1}{2},$ symmetric dispersion is stable for $\phi\in(0,\phi_{b1})$
and unstable for $\phi\in(\phi_{b1},1)$, provided that $b$ is not
too high (otherwise it is always unstable). If $b<\frac{1}{2},$ symmetric
dispersion is stable for $\phi\in(0,\phi_{b1})\cup(\phi_{b2},1)$
and unstable for $\phi\in(\phi_{b1},\phi_{b2}),$ provided that $b$
is not too low (otherwise it is always stable).
\end{prop}
\begin{proof}
    See Appendix A.4.
\end{proof}
This means that our model accounts for
the possibility of initial agglomeration as trade integration increases
from a low level, and (complete) re-dispersion once trade integration becomes high enoug.

\subsubsection{Asymmetric dispersion}

Although we cannot find an explicit stability condition for any asymmetric
dispersion equilibrium $z^{*}\in\left(\frac{1}{2},1\right)$, we can use 
equation (\ref{eq:equilibriumcondition}) that solves the equilibrium condition $\Delta v=0$ given
implicitly by $\lambda=\lambda^{*}(z)$ in the proof of Proposition
2 (Appendix A.2).\footnote{The same approach was adopted e.g. by \citet{Gaspar-et-al-ET2018} and \citep{Gaspar-et-al-RSUE2021}.} Then the stability condition of an asymmetric dispersion
equilibrium is given by:
\[
\left.\frac{d\Delta v}{dz}(z^{*})\right|_{\lambda=\lambda^{*}(z)}<0.
\]

\noindent Specifically, using (\ref{eq:indirect utility final}) and differentiating (\ref{eq:utility differential}) with respect to $z$, and evaluating at (\ref{eq:equilibriumcondition}), we get that an asymmetric equilibrium $z^{*}\in\left(\frac{1}{2},1\right)$ is stable if $\lambda^{*}(z)>0$ and:

\begin{align}
\mathcal{G}\equiv & (2z-1)\left(\phi^{2}-1\right)\left[(2b-1)\gamma(\sigma-1)(1-2z)^{2}-\sigma\right]+\nonumber \\
 & +\sigma\left[2z^{2}(\phi-1)^{2}-2z(\phi-1)^{2}+\phi^{2}+1\right]\ln\left[\frac{z(\phi-1)+1}{z(1-\phi)+\phi}\right]<0.\label{eq:stabasdisp}
\end{align}

\noindent We have the following result.

\begin{prop}
If $b<\frac{1}{2}$ an asymmetric equilibrium $z^{*}\in\left(\frac{1}{2},1\right)$ is stable for a  high enough related variety. If $b>\frac{1}{2}$, an asymmetric equilibrium is always stable when it exists.
\end{prop}

\begin{proof}
See Appendix A.
\end{proof}

\begin{figure}
\begin{centering}
\includegraphics[scale=0.75]{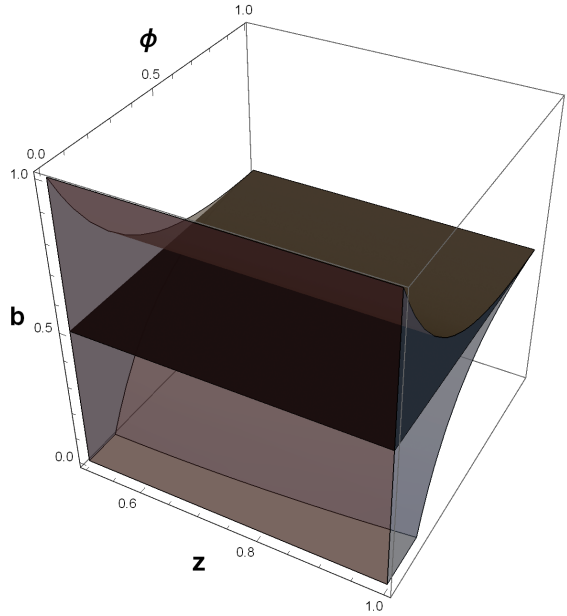}
\par\end{centering}
\caption{Stability of asymmetric dispersion. The less opaque surface corresponds to
$\mathcal{G}<0$ and $\lambda^*(z)>0$ in $(z,\phi,b)$--space for $\sigma=8$ and $\gamma=1$. The black plane corresponds
to $b=\frac{1}{2}.$ }
\label{asymmdispersionfig}
\end{figure}

 Figure~\ref{asymmdispersionfig} illustrates the Proposition by setting $\sigma=8$ and $\gamma=1$ and plotting the surface corresponding to $\left\{(z,\phi,b):\mathcal{G}<0 \cap \lambda^*(z)>0\right\}$. For $b<\frac{1}{2}$, an asymmetric equilibrium may exist that is not stable, and a higher $b$ favours its stability. If $b>\frac{1}{2}$, an asymmetric equilibrium is always stable when it exists, but its existence seems to be favoured by a lower $b$. In other words, an asymmetric equilibrium exists and is stable when $b$ is close enough to $\frac{1}{2}$. 
 
 Regarding $\phi$, a higher freeness of trade  seems to disfavour the stability of asymmetric dispersion.

\section{The impact of economic integration}

It is common in geographical economics to study the qualitative change of the spatial economy as economic integration increases. We will now look at some bifurcation diagrams using the freeness of trade, $\phi$, as the bifurcation parameter. To provide a complete gallery, we depict 6 qualitatively different scenarios, keeping most parameter values constant (except for the sixth scenario) and varying $b$, thus placing emphasis on changes in the value of related variety. Two additional illustrations for different parameter values are provided in Appendix B. The illustrations, 8 in total,  exhaust all mathematical/numerical possibilities.\footnote{This can be shown through the combination of the analysis performed in the previous sections with various simulations under a very wide range of parameter values.} The six scenarios analysed in this Section are as follows:

\begin{enumerate}[(i).]
\item  $\left(\lambda,\gamma,\sigma,b\right)=\left(2,1,8,0.33\right)$;
\item   $\left(\lambda,\gamma,\sigma,b\right)=\left(2,1,8,0.338\right)$;
\item  $\left(\lambda,\gamma,\sigma,b\right)=\left(2,1,8,0.339\right)$; 
\item  $\left(\lambda,\gamma,\sigma,b\right)=\left(2,1,8,0.35\right)$;
\item  $\left(\lambda,\gamma,\sigma,b\right)=\left(2,1,8,0.55\right)$;
\item $\left(\lambda,\gamma,\sigma,b\right)=\left(4,1,8,0.55\right)$;
\end{enumerate}

\begin{figure}
     \centering
     \begin{subfigure}{0.47\textwidth}
         \centering
         \includegraphics[width=\textwidth]{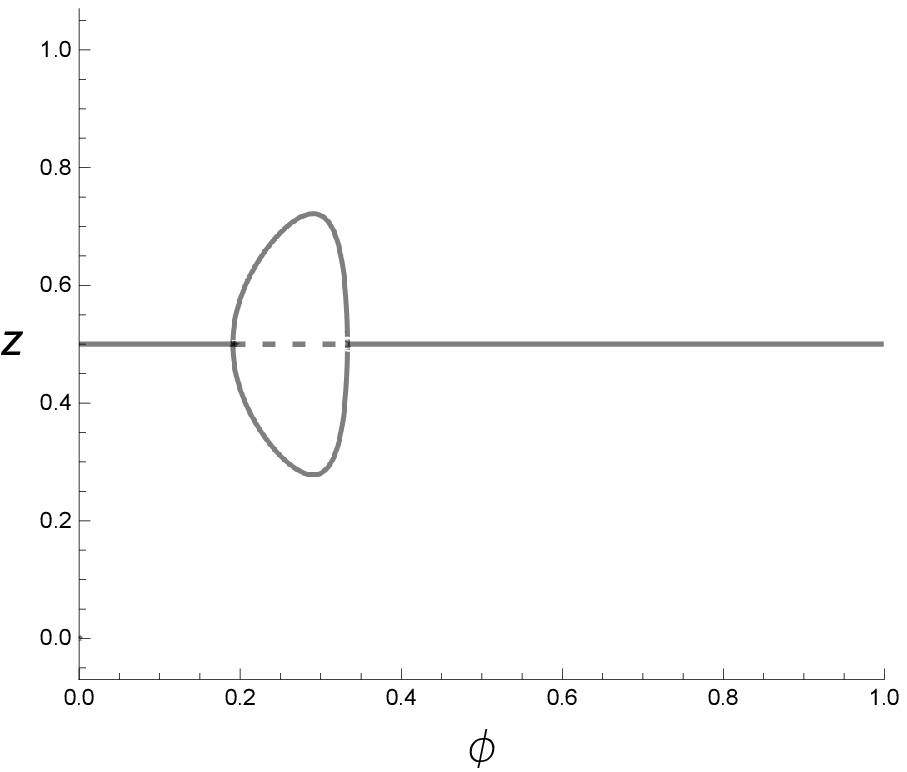}
     \end{subfigure}
     \hfill
     \begin{subfigure}{0.47\textwidth}
         \centering
         \includegraphics[width=\textwidth]{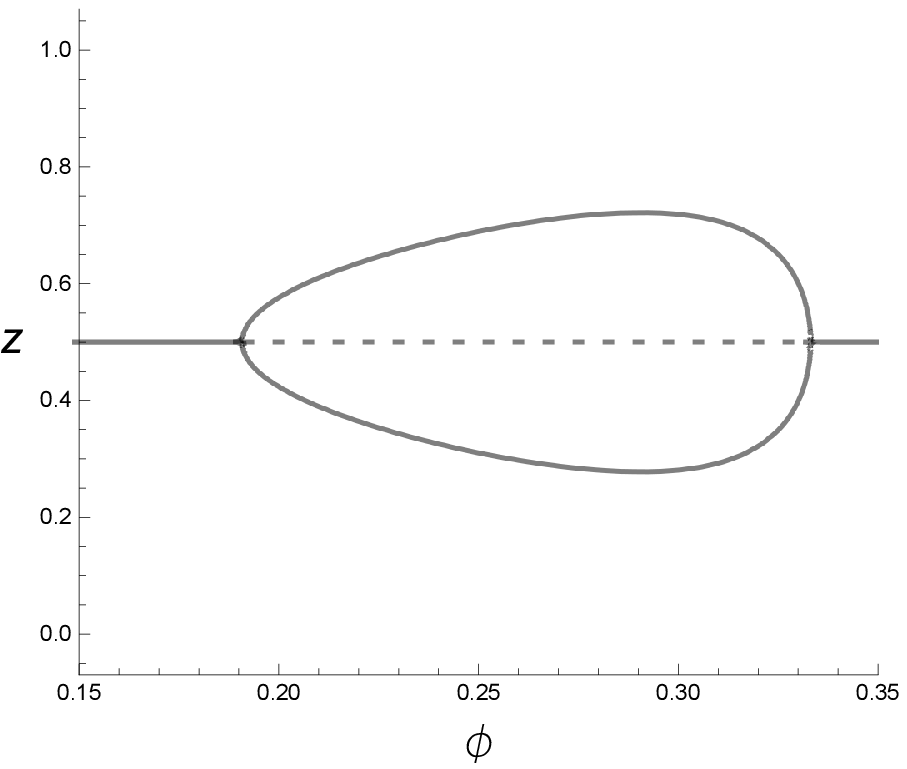}
     \end{subfigure}
     \caption{Bifurcation diagram for scenario (i). Filled lines correspond to stable equilbria and dashed lines correspond to unstable equilibria.}
     \label{bifdiagramscenarioi}
\end{figure}

\begin{figure}
     \centering
     \begin{subfigure}{0.47\textwidth}
         \centering
         \includegraphics[width=\textwidth]{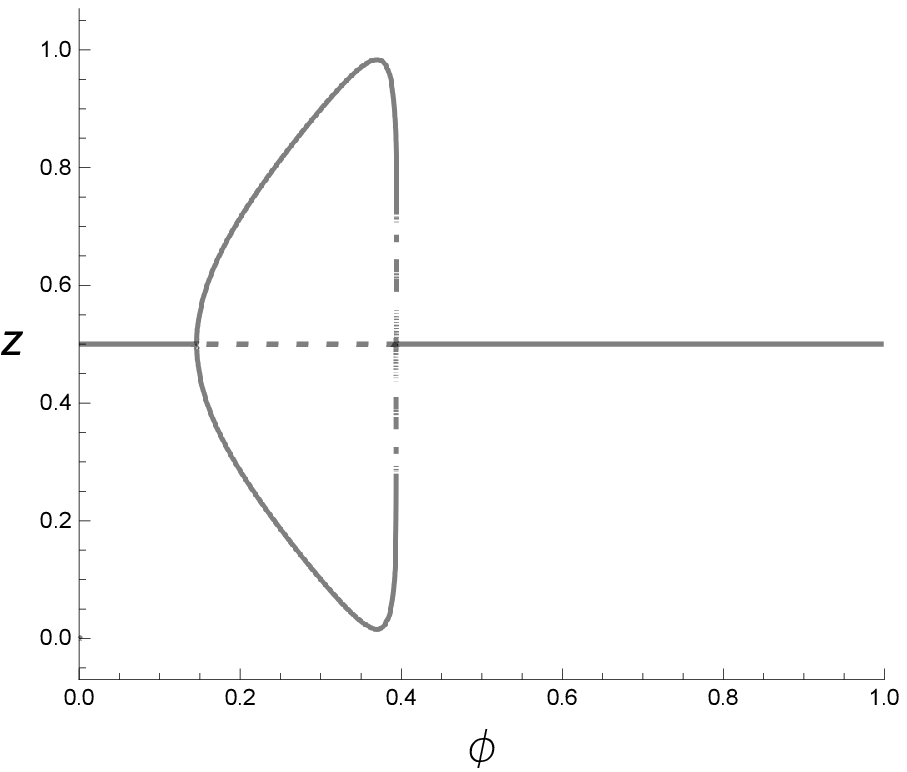}
     \end{subfigure}
     \hfill
     \begin{subfigure}{0.47\textwidth}
         \centering
         \includegraphics[width=\textwidth]{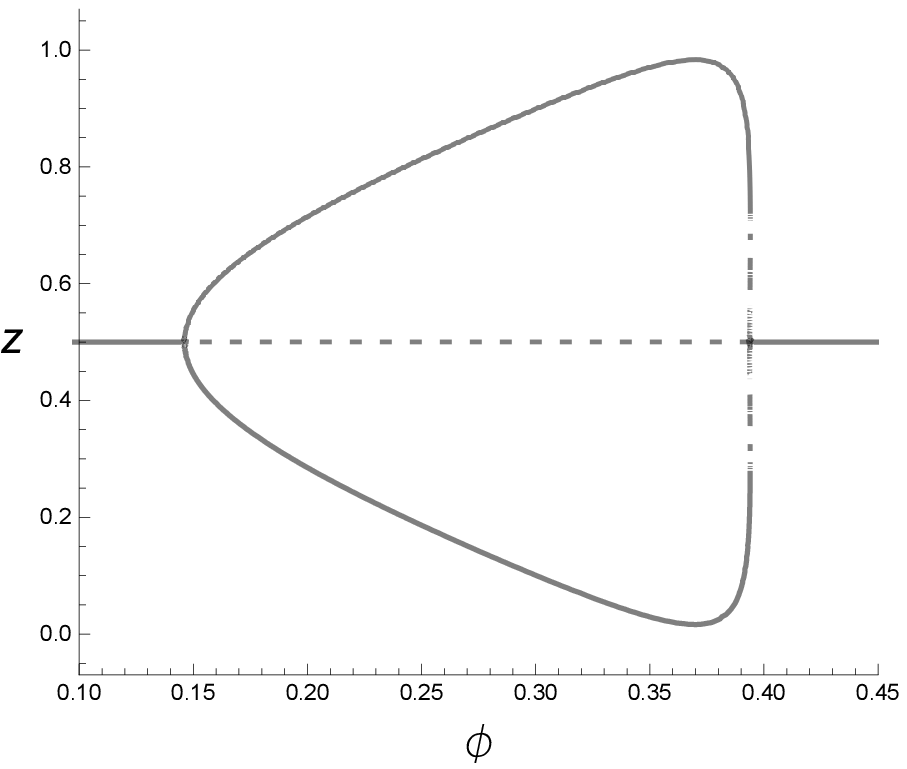}
     \end{subfigure}
     \caption{Bifurcation diagram for scenario (ii). Filled lines correspond to stable equilbria and dashed lines correspond to unstable equilibria.}
     \label{bifdiagramscenarioii}
\end{figure}

\begin{figure}
     \centering
     \begin{subfigure}{0.47\textwidth}
         \centering
         \includegraphics[width=\textwidth]{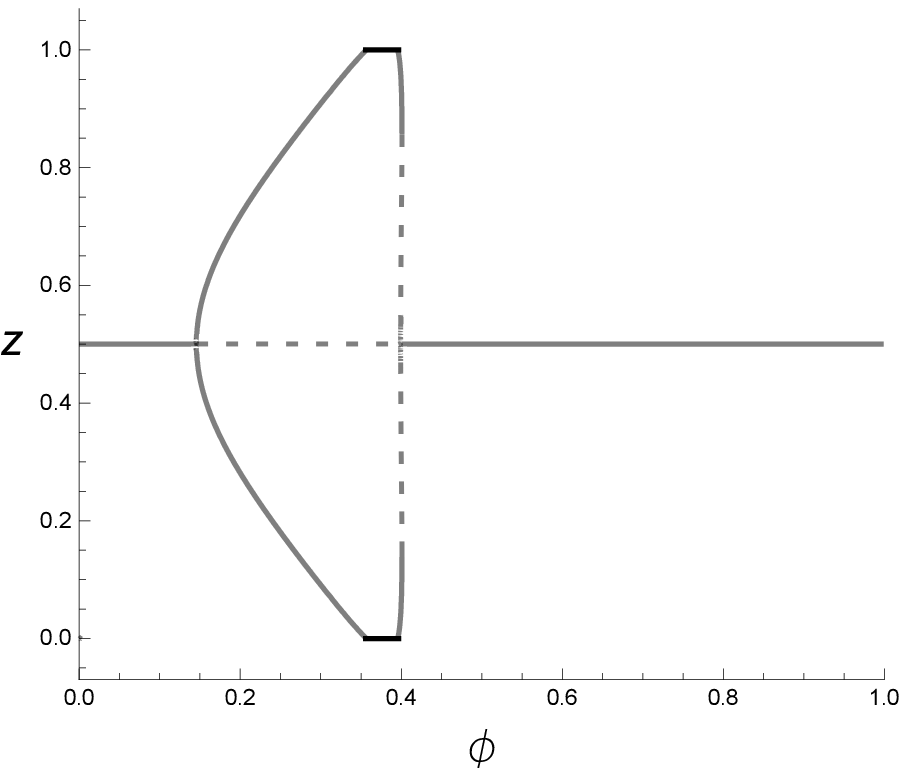}
     \end{subfigure}
     \hfill
     \begin{subfigure}{0.47\textwidth}
         \centering
         \includegraphics[width=\textwidth]{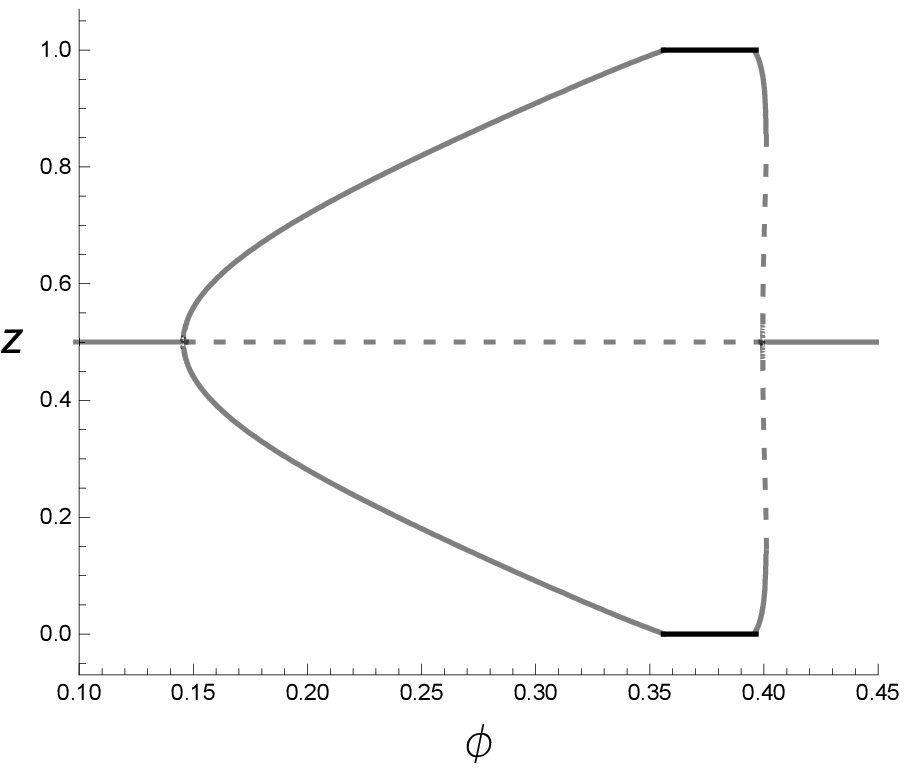}
     \end{subfigure}
     \caption{Bifurcation diagram for scenario (iii). Filled lines correspond to stable equilibria, and dashed lines correspond to unstable equilibria.}
     \label{bifdiagramscenarioiii}
\end{figure}

\begin{figure}
     \centering
     \begin{subfigure}{0.47\textwidth}
         \centering
         \includegraphics[width=\textwidth]{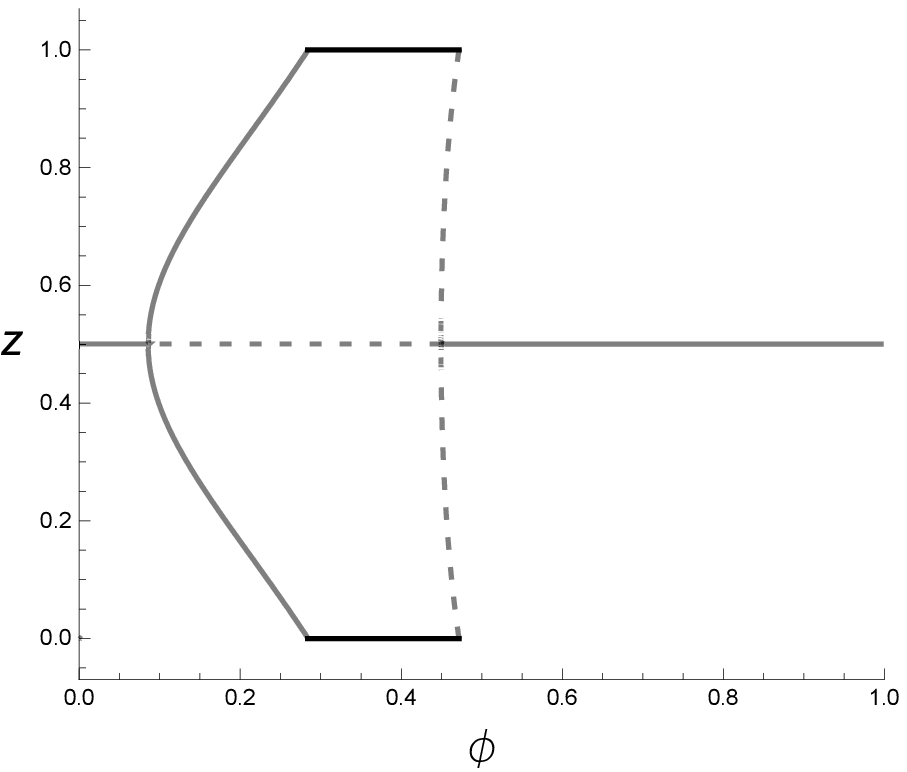}
     \end{subfigure}
     \hfill
     \begin{subfigure}{0.47\textwidth}
         \centering
         \includegraphics[width=\textwidth]{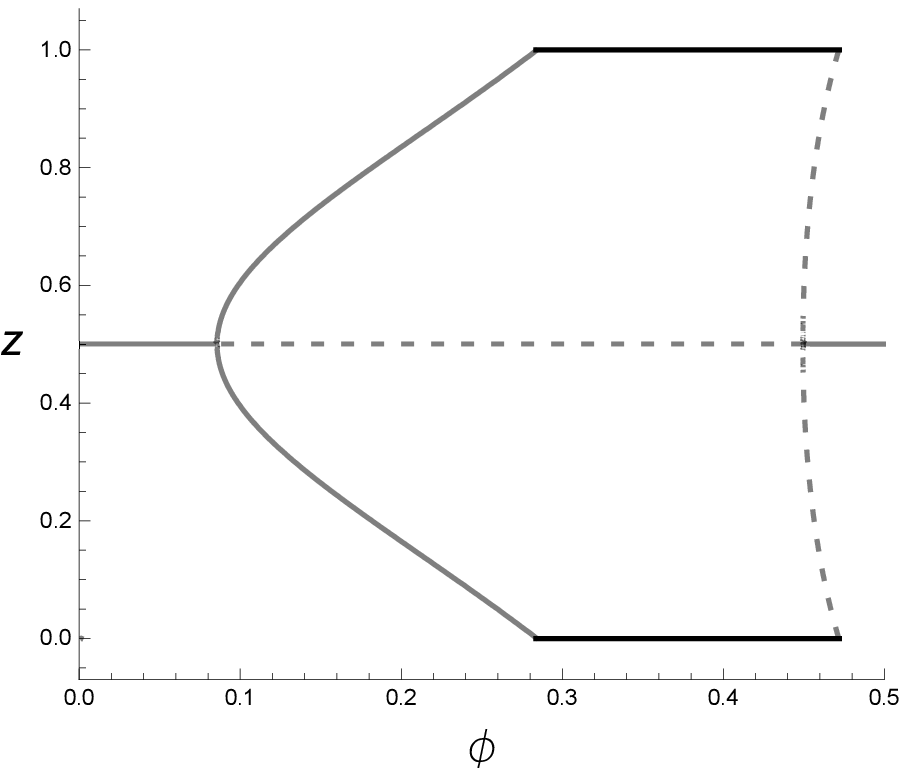}
     \end{subfigure}
     \caption{Bifurcation diagram for scenario (iv). Filled lines correspond to stable equilbria and dashed lines correspond to unstable equilibria.}
     \label{bifdiagramscenarioiv}
\end{figure}

\begin{figure}
     \centering
     \begin{subfigure}{0.47\textwidth}
         \centering
         \includegraphics[width=\textwidth]{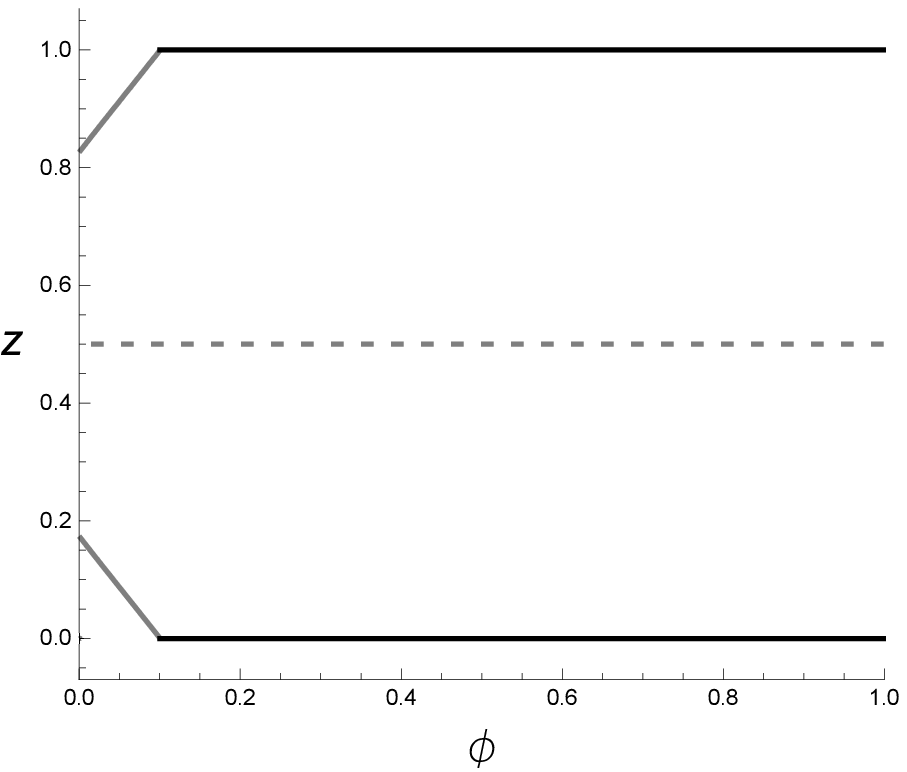}
     \end{subfigure}
     \hfill
     \begin{subfigure}{0.47\textwidth}
         \centering
         \includegraphics[width=\textwidth]{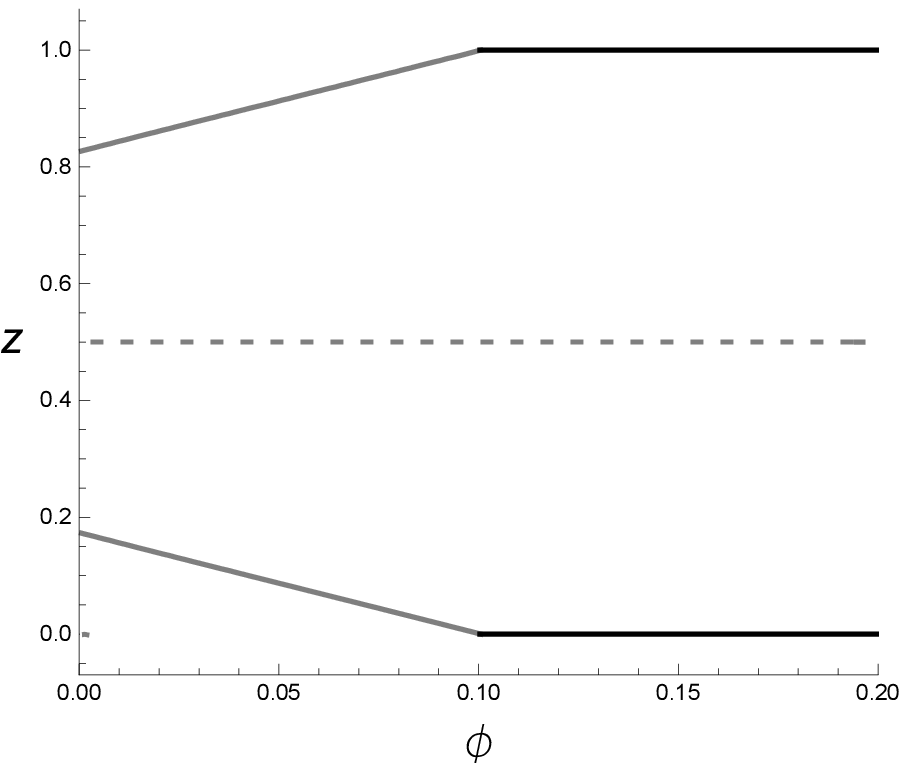}
     \end{subfigure}
     \caption{Bifurcation diagram for scenario (v). Filled lines correspond to stable equilbria and dashed lines correspond to unstable equilibria.}
     \label{bifdiagramscenariov}
\end{figure}

\begin{figure}
     \centering
     \begin{subfigure}{0.47\textwidth}
         \centering
         \includegraphics[width=\textwidth]{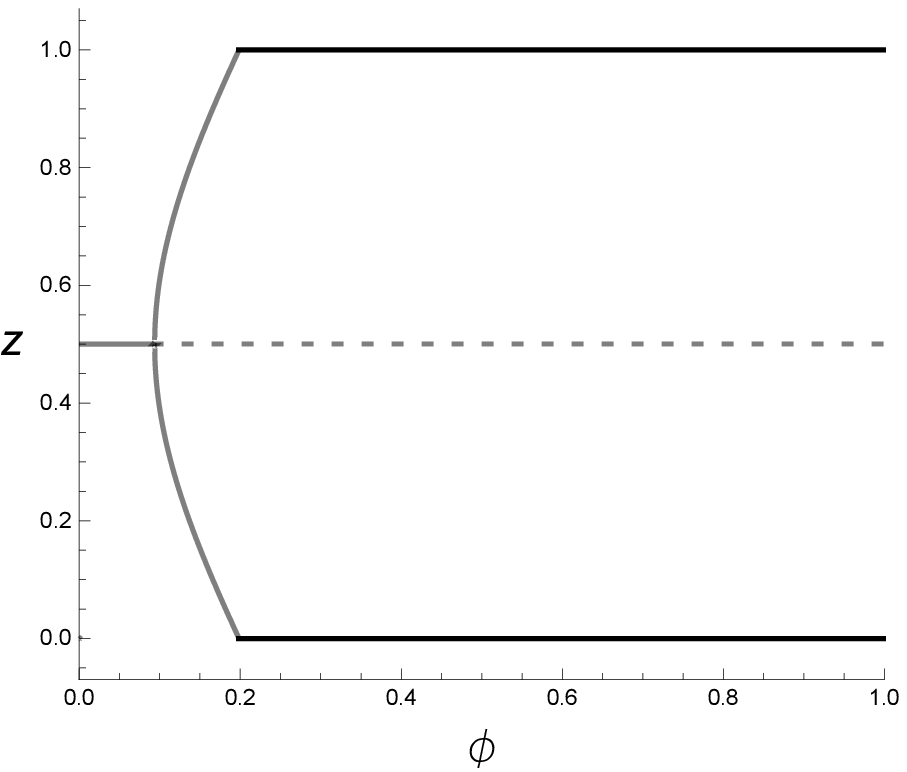}
     \end{subfigure}
     \hfill
     \begin{subfigure}{0.47\textwidth}
         \centering
         \includegraphics[width=\textwidth]{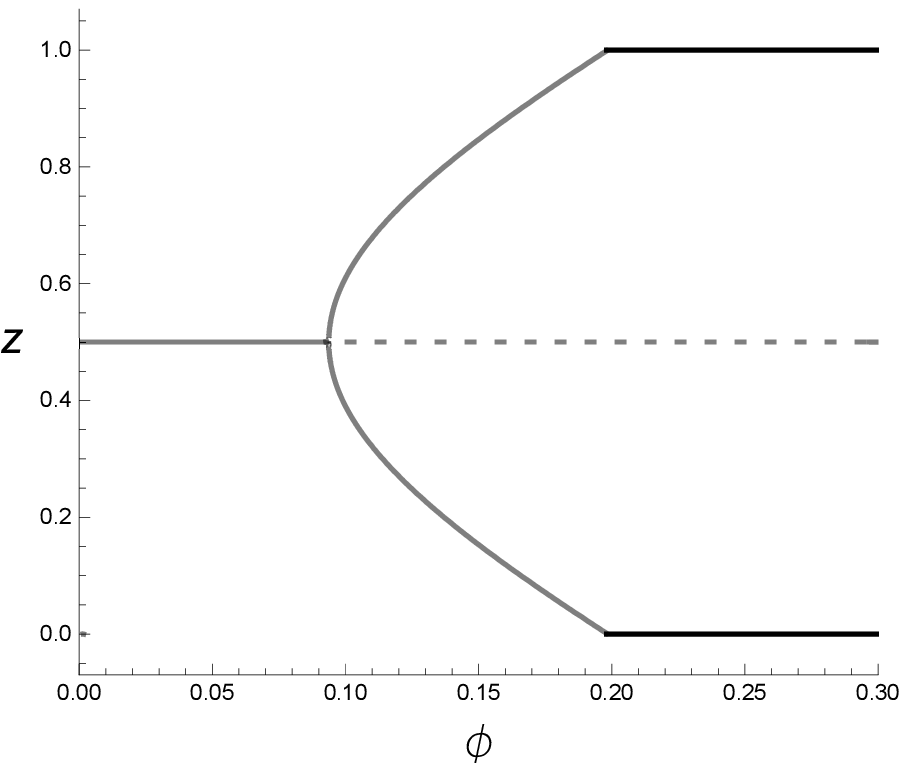}
     \end{subfigure}
     \caption{Bifurcation diagram for $(\lambda,\gamma,\sigma,b)=(4,0.9,8,0.55)$. Filled lines correspond to stable equilbria and dashed lines correspond to unstable equilibria.}
     \label{bifdiagramscenariobhigh}
\end{figure}

For a prohibitively low related variety, scientists disperse evenly among the two regions, irrespective of the value of the freeness of trade. The economic intuition is simple: a lower related variety implies higher chance of successful innovation with more scientists living in the other region. Hence, the nominal wage is higher when the scientists are more evenly distributed. 

In scenario (i), shown in Figure~\ref{bifdiagramscenarioi}, related variety is such that within-region interaction is relatively less important ($b=0.33$). For a low freeness of trade, symmetric dispersion is stable because firms wish to avoid the burden of a very costly transportation supplying to farmers from full agglomeration in a single region. As $\phi$ increases, the economy initially agglomerates, but then re-disperses as $\phi$ increases further. This re-dispersion process occurs because, for a very high economic integration, firms find it profitable to relocate to the less industrialized  region in order to benefit from the pool of scientists in the more agglomerated region, which generates a higher chance of innovation and thus higher expected profits. Noteworthy, the turning point in the agglomeration process happens \textit{before} industry reaches full agglomeration in a single region, as in \citet{Pfluger-Suedekum-JUE2008}. However, contrary to the latter, our model does not predict full agglomeration in the entire parameter range of economic integration when related variety is low enough. Re-dispersion in scenario (i) is more akin to geographical economic models of vertical linkages between upstream and downstream firms by \citet{Krugman-Venables-QJE1995,venables1996equilibrium} and \citet{Puga-EER1999}. However, in these models, re-dispersion is smooth altogether and occurs when workers are inter-regionally \emph{immobile} and firms become too sensitive to regional cost differentials when economic integration is very high.

In scenario (ii), illustrated by Figure~\ref{bifdiagramscenarioii}, related variety is just slightly higher, and the model still accommodates for re-dispersion. However, the re-dispersion process is not smooth -- the economy suddenly jumps to symmetric dispersion from a fairly  asymmetric equilibrium spatial distribution. 

Scenario (iii) also just slightly increases related variety compared to the previous scenario (see Figure~\ref{bifdiagramscenarioiii}), and the story of spatial outcomes as economic integration increases is very similar, except that, in this case, full agglomeration is stable for a small range of intermediate values of $\phi$, as predicted by Proposition 3. The parametrization here  also corresponds to that illustrated in Figure~\ref{fig:Figure1}. 

The re-dispersion processes of scenarios (ii) and (iii) are  uncommon in the literature of geographical economics; rather, such jumps occur in early models \citep{Fujita-Krugman-Venables-Book1999,BaldwinForslidMartinOttavianoRobertNicoud+2003} from the state of symmetric dispersion to catastrophic agglomeration \citep{behrens2011tempora}. 
The reverse discontinuous jump, i.e., from symmetric dispersion to partial agglomeration as trade costs steadily decrease, has been uncovered in the model by \citet{pfluger2010size}, where \emph{all} production factors, except land, which is used both for housing and production, are inter-regionally mobile. Their conclusions about spatial outcomes  reveal a line-symmetry of scenario (iii): as integration increases, the economy jumps discontinuously from symmetric dispersion to partial agglomeration, and
the ensuing re-dispersion is gradual and continuous.\footnote{In our model, the assumption that unskilled workers are immobile is useful for tractability, as is the case of all footloose entrepreneur models \citep{BaldwinForslidMartinOttavianoRobertNicoud+2003}. However, we make the reasonable conjecture that immobile labour generates an unnecessary dispersion force that  changes the conclusions of our model compared to the case of a perfectly mobile workforce only in the sense of ``reversed'' stability as transport costs decrease, i.e. the line-symmetry of all scenarios (i)--(vi). }

Figure~\ref{bifdiagramscenarioiv} illustrates scenario (iv) and shows that the sudden re-dispersion process under a slightly higher $b$ now happens from the state of full agglomeration directly to the state of symmetric dispersion. In both scenarios (iii) and (iv), the state of agglomeration is stable for intermediate values of economic integration, as in \citet{robert2008offshoring}.

In scenario (v), for a sufficiently high related variety ($b>1/2$), within-region interaction among scientists improves the chances of innovation enough such that the real wage becomes higher when they are either partially agglomerated in one region for low values of $\phi$, or completely agglomerated in one region for a high enough $\phi$. This is portrayed in Figure~\ref{bifdiagramscenariov}. Scenario (v) precludes the so-called ``no black-hole condition'' \citep{Fujita-Krugman-Venables-Book1999}, a condition that the constant elasticity of substitution $\sigma$ must be high enough such that symmetric dispersion can be stable for low enough economic integration. As argued by \citet{Gaspar-et-al-ET2018}, this condition may be unwarranted if its exclusion allows for spatial outcomes other than ubiquitous agglomeration. 

We can thus conclude that a higher related-variety is associated with a more pronounced agglomeration during the industrialization process, for intermediate values of economic integration, until eventually it becomes so high that re-dispersion is no longer possible because within-region interaction among scientists is too important to make any deviation to a deindustrialized region worthwhile. 

In scenario (vi) we illustrate the qualitative change in the spatial structure of the economy as $\phi$ increases for $b>1/2$, but  with a higher $\lambda$, since, with the parameter values of the previous scenario, agglomeration would be ubiquitously stable (and hence uninteresting) for higher values of $b$. In Figure~\ref{bifdiagramscenariobhigh}, we can observe a supercritical pitchfork bifurcation, as in the model by \cite{pfluger2004simple}, where there is no innovation. That is, for low levels of economic integration, symmetric dispersion is stable. As $\phi$ increases, one region smoothly becomes more and more industrialized en route to a full agglomeration whereby that region becomes a core.

Noteworthy, in any scenario for which a break point exists at symmetric dispersion, it can be shown analytically that the model undergoes a pitchfork bifurcation, either supercritical  or subcritical, depending on the parameter values (see Appendix A.6). Additionally, in Figures~\ref{bifdiagramscenarioii} and \ref{bifdiagramscenarioiii} (scenarios (ii) and (iii)), a limit point $\phi \equiv \phi_l \in (\phi_{b2},1)$ is discernible at which two asymmetric equilibria, along a curve tangent to $\phi_l$ that lies to its left, collide and coalesce. This suggests that in both scenarios (i) and (ii) the model undergoes a saddle-node bifurcation at some asymmetric equilibrium $z^{*}\in\left(\frac{1}{2},1\right)$. This kind of bifurcation also appears in the two-region footloose entrepreneur model by \cite{Forslid-Ottaviano-JEG2003} with heterogeneous agents analysed by \cite{Castro_2021} and also in the \cite{pfluger2004simple} model extended to multiple regions by \cite{Gaspar-et-al-ET2018}. This kind of bifurcation seems to be associated with discontinuous jumps between some asymmetric equilibrium other than agglomeration and the symmetric dispersion once $\phi$ rises (falls) above (below) some threshold level.

In Appendix B we change the benchmark parameter values and further illustrate how changes in $b$ can bring about richer implications regarding the qualitative structure of the spatial economy.


 
 
\section{On the role of regional interaction}

\subsection{The impact of related variety: comparative statics}

It is worthwhile investigating analytically how changes in the level of related variety  affect the long-run spatial outcomes in the economy. Following \cite{Castro_2021}, we say that a change in related variety \emph{favours agglomeration} if, due to the change: (a) symmetric dispersion may become unstable but not stable, (b) agglomeration may become stable but not unstable, and (c) asymmetric dispersion becomes more asymmetric.

 Using (\ref{eq:indirect utility final}) and (\ref{eq:utility differential}), we have:
\[
\frac{\partial\Delta v}{\partial b}=\frac{\gamma\mu(2z-1)(\phi+1)\left[\lambda-4z^{2}+\phi(\lambda+4(z-1)z+2)+4z\right]}{2\sigma\left[z(\phi-1)+1\right]\left[z(1-\phi)+\phi\right]},
\]
which is positive for $z\in\left(\frac{1}{2},1\right)$.    We have the following result.

\begin{lem}
An increase in related variety favours agglomeration.
\end{lem}
\begin{proof}
See Proposition 9 of \cite[p.197]{Castro_2021}.
\end{proof}

The interpretation behind this result is straightforward: a higher $b$ implies a higher chance of successful innovation in region $i$ when more scientists live in region $i$. Hence, expected profits become higher, which leads to higher wages and, thus, a higher utility differential in region $i$. Therefore, a higher $b$ makes stability of agglomeration (symmetric dispersion) more (less) likely for a given value of $\phi$, and  asymmetric dispersion becomes more asymmetric. 

Suppose $b$ is such that symmetric dispersion is the unique stable equilibrium for a low $\phi$ (Proposition 4), asymmetric dispersion is the unique stable equilibrium for intermediate values of $\phi$ (Proposition 5), and agglomeration is the unique stable equilibrium for a high $\phi$ (Proposition 3). Accordingly, the model undergoes a supercritical pitchfork bifurcation at symmetric dispersion. As illustrated in Figure~9, an increase in $b$ shifts the pitchfork bifurcation leftwards, provided that the set of stable equilibria remains unchanged.


\begin{figure}[h]

\begin{centering}
\includegraphics{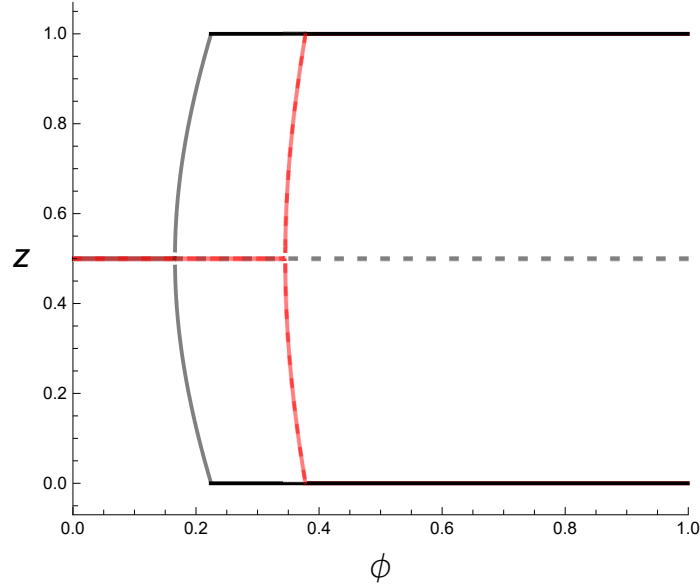}\caption{Bifurcation diagrams: (i) to the left, in black, we set $b=0.9$;
(ii) to the right, in red, we set $b=0.7$. Parameter values are $(\sigma,\gamma,\lambda)=(8,1,4).$}
\par\end{centering}
\end{figure}

\subsection{Related variety and re-dispersion: a general case}

Let us now consider a more general specification for the firms' success
of innovation:
\begin{equation}
\Phi_{i}^{G}(s)=\min\left\{ \frac{g_{i}(z_{i},z_{j};v)\gamma A}{a_{1}(s)},1\right\} ,\label{eq:general}
\end{equation}
with $g_{1}\equiv g(z)$, $g_{2}\equiv g(1-z)$, $g:[0,1]\mapsto[0,1]$
is a $C^{3}$ function of $z$, and $v\in\mathcal{V}$ is a vector
of parameters which includes the related variety parameter $b\in(0,1)$.
Additionally, assume that $g^{\prime}(z)<0$ if $b<1/2$, and $g^{\prime}(z)>0$
if $b>1/2$, for $z\in[\frac{1}{2},1]$. This reflects the fact that $g(z)$ constitutes a dispersion
force when related variety is low, and an agglomeration force when
related variety is high, \emph{at least} when region $1$ is either the same size or larger than region $2$. This less restrictive assumption might be more reasonable in some contexts, e.g., if congestion is assumed to have an impact on knowledge creation (either positive or negative). Clearly, $\Phi_{i}(s)$ in (\ref{eq:prob of succesful innovation-1}) is a particular
case of $\Phi_{i}^G(s)$ in (\ref{eq:general}). The indirect utility
in region $i$ is obtained by replacing the term $bz_{i}+(1-b)z_{j}$
in (\ref{eq:indirect utility final}) with $g_{i}(z)$. 

At symmetric dispersion, making use of the fact that $g_1^{\prime}(\tfrac{1}{2})=-g_2^{\prime}(\frac{1}{2}),$
we have at most two break points:
\begin{align*}
\phi_{b1}^{G} & =\frac{\gamma(\lambda+1)(\sigma-1)\left[g'\left(\frac{1}{2}\right)+2g\left(\frac{1}{2}\right)\right]-\kappa}{2\left[\gamma g\left(\frac{1}{2}\right)(\lambda+2)(\sigma-1)+\sigma\right]-\gamma(\lambda+1)(\sigma-1)g'\left(\frac{1}{2}\right)},\\
\phi_{b2}^{G} & =\frac{\gamma(\lambda+1)(\sigma-1)\left[g'\left(\frac{1}{2}\right)+2g\left(\frac{1}{2}\right)\right]+\kappa}{2\left[\gamma g\left(\frac{1}{2}\right)(\lambda+2)(\sigma-1)+\sigma\right]-\gamma(\lambda+1)(\sigma-1)g'\left(\frac{1}{2}\right)},
\end{align*}
where 
\[
\kappa=2\sqrt{2\gamma^{2}g\left(\frac{1}{2}\right)(\lambda+1)^{2}(\sigma-1)^{2}g'\left(\frac{1}{2}\right)+\left[\gamma g\left(\frac{1}{2}\right)(\sigma-1)+\sigma\right]^{2}}.
\]
Since $g(\frac{1}{2})>0,$ it can be shown that $\phi_{b2}^{G}\notin(0,1)$
if $g^{\prime}(\frac{1}{2})>0,$ i.e., if $b>1/2$. Since symmetric
dispersion is stable for $\phi\in(0,\phi_{b1}^{G})\cup(\phi_{b2}^{G}),$
we can thus conclude that the complete re-dispersion of economic activities is
not possible if related variety is high.  This leads to the following result.
\begin{prop}
Suppose region $i$ is the same size or larger than region $j$. If intra-regional interaction is relatively more important compared
to inter-regional interaction for the success of innovation in region $i$, a complete re-dispersion
of economic activities for high economic integration is not possible.
\end{prop}
In other words, the process
of (complete) re-dispersion depends on the dispersive or agglomerative nature
of knowledge spillovers at symmetric dispersion.

We take one step further by investigating what type of bifurcation
the symmetric dispersion undergoes at some break point $\phi_{b}\in\left\{ \phi_{b1}^{G},\phi_{b2}^{G}\right\} $.
We have the following derivatives:
\[
\dfrac{\partial f}{\partial z}\left(\dfrac{1}{2};\phi_{b}\right)=0;\ \dfrac{\partial^{2}f}{\partial z^{2}}\left(\dfrac{1}{2};\phi_{b}\right)=0;\ \dfrac{\partial f}{\partial\phi}\left(\dfrac{1}{2};\phi_{b}\right)=0;
\]
Furthermore, we have:
\[
\dfrac{\partial^{2}f}{\partial\phi\partial z}\left(\dfrac{1}{2};\phi_{b}\right)=\frac{8\mu}{(\phi+1)^{3}}\left\{ \frac{\gamma g\left(\frac{1}{2}\right)\left[-2\lambda(\phi_{b}-1)-3\phi_{b}+1\right]}{\sigma}+\frac{\phi_{b}+1}{1-\sigma}\right\} ,
\]
which is zero if and only if $g(\frac{1}{2})=g_{b},$ with
\[
g_{b}\equiv-\frac{\sigma(\phi_{b}+1)}{\gamma(\sigma-1)\left[2\lambda(\phi_{b}-1)+3\phi_{b}-1\right]}.
\]
If $g(\frac{1}{2})\neq g_{b},$ then symmetric dispersion undergoes
a pitchfork bifurcation (see Appendix A.6). In this case, its criticality is determined
by the sign of the third-order derivative $\frac{\partial^{3}f}{\partial z^{3}}(\frac{1}{2};\phi_{b}).$ However, without additional details on the shape of $g(z)$ it is
impossible to convey further information on these conditions and to
relate them with $b$.

For the purpose of illustration, Appendix C depicts an alternative scenario compared to our benchmark case, whereby the regional interaction between scientists determining the success of innovation is of multiplicative nature, but still satisfies all the conditions of the general form of $g(z)$ presented in this Section. It  corroborates our general findings regarding the relation between related variety and the stability of symmetric dispersion for both very high and very low values of the trade freeness.  However, the transition between different long-run spatial distributions need not be the same. 

We conclude that, while the presence of a (complete) re-dispersion phase should depend solely on the dispersive or agglomerative nature of technological spillovers, the shift between distinct spatial outcomes due to economic integration may display differences across various functional forms describing the likelihood of a successful innovation.

 \section{Concluding remarks}
 In this paper, we have analyzed a two-region model with vertical innovations that
enhance the quality of varieties of   horizontally differentiated
manufactures produced in each of the two regions. We
looked at how the spatial creation and diffusion of knowledge and increasing
returns in manufacturing interact to shape the spatial economy. Knowledge levels translate in a firms' capacity to innovate. The chance of successful innovations depends on the spatial distribution of mobile agents in the economy, i.e. on the intra- and inter-regional
interaction between researchers (mobile workers).   

We find that, if  the weight of interaction with foreign scientists is relatively more important for the success of innovation, the model accounts for re-dispersion of economic activities after an initial stage of progressive agglomeration as trade integration increases from a low level. However, the relationship between economic integration and spatial imbalances is far from trivial, as we have shown a myriad of different qualitative possibilities regarding transitions between different stable states that depend on the weight of intra-regional interaction between scientists, i.e. on the level of related variety \citep{Frenken-etal-RS2007}. We show that the re-dispersion process is only smooth \citep{fujitathisse2013} when related variety is low. If the related variety is intermediate, there is a discontinuous jump toward symmetric dispersion from either full agglomeration or partial agglomeration. If related variety is  high, re-dispersion is precluded and full agglomeration is the only stable outcome for high enough economic integration, as in most early ``new economic geography'' models \citep{Fujita-Krugman-Venables-Book1999,BaldwinForslidMartinOttavianoRobertNicoud+2003,pfluger2004simple}.

  A note on the functional form for the chance of successful innovations is warranted. If it preserves the idea that congestion dampens innovation, then it introduces a dispersion force, since researchers have incentives to relocate to smaller regions and benefit from the sizable pool of agents living in the largest region. In this case, re-dispersion after an initial phase of agglomeration is possible. In Appendix C, we provide insights on a multiplicative scenario (as in e.g.  \citealp{Berliant-Fujita-RSUE2012}). Both the additive  and multiplicative functional forms  presented here share the main idea that public knowledge transfers imperfectly across space, as argued by e.g. \citet{krugman1991geography}, or \citet{audretsch2004knowledge}, and supported by the empirical evidence of  \citet{audretsch1996r}. However, the additive case bears richer and more plausible implications, while preserving enough analytical tractability. Nonetheless, it could be worthwhile to investigate further how different types of regional interaction determining innovation success affect spatial outcomes in the economy. Although the existence of a re-dispersion phase should hinge solely on the dispersive/agglomerative nature of technological spillovers, the transition between different spatial outcomes as a consequence of economic integration may exhibit variations across different functional forms for the probability of a successful innovation.

  Finally, we reiterate that  the modeling of the innovation sector in this paper purposefully abstracts from the dynamic processes that usually drive innovation and the creation and diffusion of knowledge. Such abstraction is useful because it allows the outcomes of the model to be entirely attributed to the spatial mechanism of knowledge spillovers (as proposed by \citet{BS2022}). However, we argue that this mechanism could be extended to models of Schumpeterian growth, where the explicit modelling of innovation dynamics would allow to infer about an eventual circular causality between regional growth and agglomeration patterns \citep{baldwin2004agglomeration}.   
 

\clearpage{}

\appendix

\section{Proofs}

This appendix contains the more cumbersome formal proofs that support
our main results.

\subsection{Proof of proposition 1}

Differentiating $\Delta v(z)$ in (\ref{eq:utility differential})
yields:
\begin{equation}
\frac{d\Delta v}{dz}(z)=\frac{\mu P(z)}{2(\sigma-1)\sigma\left[z(\phi-1)+1\right]{}^{2}\left[z(1-\phi)+\phi\right]{}^{2}},\label{prop 1 eq}
\end{equation}
where:
\[
P(z)=a_{1}z^{4}+a_{2}bz^{3}-2(1-\phi)a_{3}z^{2}+2(1-\phi)a_{4}z+a_{5},
\]
with:{\small{}
\begin{align*}a_{1}= & 4(1-2b)\gamma\mu(\sigma-1)(\phi-1)^{3}(\phi+1)\\
a_{2}= & 8(2b-1)\gamma\mu(\sigma-1)(\phi-1)^{3}(\phi+1)\\
a_{3}= & \gamma(\sigma-1)\left\{ b(\phi+1)\left[(\lambda-2)\phi^{2}-\lambda+18\phi-4\right]-\phi\left[\lambda(\phi-1)\phi+\lambda+6\phi\right]+\lambda-8\phi+2\right\} \\
 & +\sigma(\phi+1)(\phi-1)^{2}\\
a_{4}= & \gamma(\sigma-1)\left\{ \lambda(\phi-1)\left[b(\phi+1)^{2}-\phi^{2}-1\right]+2\phi\left[b(\phi+1)(\phi+5)-\phi(\phi+2)-3\right]\right\} \\
 & +\sigma(\phi+1)(\phi-1)^{2}\\
a_{5}= & \gamma(\sigma-1)\left\{ \lambda\left(\phi^{2}+1\right)\left[b(\phi+1)^{2}-\phi^{2}-1\right]+2\phi\left[b\left(\phi^{3}+3\phi^{2}+\phi-1\right)-\phi\left(\phi^{2}+\phi+1\right)+1\right]\right\} \\
 & -2\sigma\phi\left(\phi^{2}-1\right).
\end{align*}
}{\small\par}

\noindent The denominator of (\ref{prop 1 eq}) is positive, which
means that the sign of $\frac{d\Delta v}{dz}(z)$ is given by the
sign of $P(z)$, which is a fourth degree polynomial in $z$. Therefore,
$\Delta v(z)$ has at most four turning points, and thus at most
five equilibria for $z\in[0,1]$. We know that $z=\frac{1}{2}$ is
an invariant pattern. By symmetry, we can establish that there exist
at most two equilibria for $z>\tfrac{1}{2}$, which concludes the
proof.\hfill{}$\square$

\subsection{Proof of Proposition 2}

Proceeding in a familiar fashion as in \citet{Gaspar-et-al-ET2018,Gaspar-et-al-RSUE2021}, the equilibrium condition $\Delta v(z)=0$ yields:
\begin{equation}
\lambda\equiv\lambda^{*}(z)=-2\frac{b_{1}b_{2}+b_{3}\ln\left[\frac{z(\phi-1)+1}{z(1-\phi)+\phi)}\right]}{b_{4}},\label{eq:equilibriumcondition}
\end{equation}
where:
\begin{align*}
b_{1} & =\gamma(\sigma-1)(2z-1)\\
b_{2} & =\phi^{2}\left[2b(z-1)z+b-z^{2}+z-1\right]+(1-2b)(z-1)z+b\phi\\
b_{3} & =\sigma\left[z(\phi-1)+1\right]\left[z(\phi-1)-\phi\right]\\
b_{4} & =\gamma(\sigma-1)(2z-1)\left[b(\phi+1)^{2}-\phi^{2}-1\right].
\end{align*}
It is easy to note that $\lambda^{*}(z)$ has a vertical asymptote
if and only if $b_{4}=0$, i.e., iff:
\[
b=\hat{b}\equiv\frac{\phi^{2}+1}{(\phi+1)^{2}}.
\] 

\noindent For $z\in\left(\frac{1}{2},1\right],$ the log term of $\lambda^{*}(z)$
is negative, as is $b_{3}$. Next, we have $b_{1}>0$, and $b_{2}>0$
if:
\[
b \geq \underline{b}\equiv\frac{(z-1)z\left(\phi^{2}-1\right)+\phi^{2}}{2(z-1)z\left(\phi^{2}-1\right)+\phi(\phi+1)},
\]
where $0<\underline{b}<\hat{b}$. Since $b_{4}<0$ only if $b<\hat{b}$,
we have $\lambda^{*}(z)>0$ if $b\in\left[\underline{b},\hat{b}\right)$
and $\lambda^{*}(z)<0$ if $b\in\left(\hat{b},1\right)$. For $b\in\left(0,\underline{b}\right)$,
we need further inspection.

We have that:{\footnotesize{}
\[
\frac{\partial\lambda^{*}}{\partial b}(z)=\frac{2(\phi+1)\left[z(\phi-1)+1\right]\left[z(\phi-1)-\phi\right]\left\{ \gamma(\sigma-1)(2z-1)(\phi-1)+\sigma(\phi+1)\ln\left[\frac{z(\phi-1)+1}{z(1-\phi)+\phi)}\right]\right\} }{\gamma(\sigma-1)(2z-1)\left[-b(\phi+1)^{2}+\phi^{2}+1\right]^{2}},
\]
}which is positive for all $z\in\left(\frac{1}{2},1\right].$ The unique zero  of $\lambda^{*}(z)$ in terms of $b$ is given by:{\footnotesize{}
\[
b=\tilde{b}\equiv\frac{\gamma(\sigma-1)(2z-1)\left[(z-1)z\left(\phi^{2}-1\right)+\phi^{2}\right]-\sigma\left[z(\phi-1)+1\right]\left[z(\phi-1)-\phi\right]\ln\left[\frac{z(\phi-1)+1}{z(1-\phi)+\phi)}\right]}{\gamma(\sigma-1)(2z-1)(\phi+1)\left[2(z-1)z(\phi-1)+\phi\right]},
\]
}
with $\tilde{b} < \underline{b}$ and $\lambda^*(z)>0$ for $b\in (\tilde{b},\hat{b})$. It is possible to show that $\tilde{b}$ is increasing in $\gamma.$
Moreover, we have $\tilde{b}=0$ if and only if:
\[
\gamma=\gamma_{c}\equiv\frac{\sigma\left[z(1-\phi)-1\right]\left[z(1-\phi)+\phi\right]\ln\left[\frac{z(\phi-1)+1}{z(1-\phi)+\phi}\right]}{(\sigma-1)(2z-1)\left[(z-1)z\left(\phi^{2}-1\right)+\phi^{2}\right]}>0.
\]
This means that $\tilde{b}\geq0$ if  $\gamma\geq\gamma_{c}$
and $\tilde{b}<0$ if $\gamma<\gamma_{c}$ and $\gamma_{c}\in(0,1]$.
Since $\gamma_{c}\in(0,+\infty)$, we have $\tilde{b}<0$ if $\gamma_{c}>1$.
As a result, we have $\tilde{b}<0$ if $\gamma\in\left(0,\min\{1,\gamma_{c}\}\right)$ and $\tilde{b}\geq 0$ if $\gamma\in\left[\min\{1,\gamma_{c}\},1\right)$.
Then $\lambda^{*}(z)>0$ if $\gamma\in\left(0,\min\{1,\gamma_{c}\}\right)$
and $b\in(0,\hat{b})$. Otherwise, we have $\lambda^{*}(z)>0$ if
$\gamma\in\left[\min\{1,\gamma_{c}\},1\right)$ and $b\in(\tilde{b},\hat{b}).$ Therefore, $\lambda^{*}(z)$
is positive for $b\in\left(\max\left\{ 0,\tilde{b}\right\} ,\hat{b}\right)$
and negative for $b\in\ \left(0,\max\left\{ 0,\tilde{b}\right\}\right)\cup\left(\hat{b},1\right)$, where $\max\left\{ 0,\tilde{b}\right\}$ depends on $\gamma_c$ and on the value of $\gamma$ as described above.

Thus, we can
assert that, if $b\in\left(\max\left\{ 0,\tilde{b}\right\} ,\hat{b}\right)$, there exists a value of $\lambda>0$ such that at least one (at most two) dispersion
equilibrium $z\equiv z^{*}\in\left(\frac{1}{2},1\right]$ exists. This concludes the proof.\hfill{}$\square$

\subsection{Proof of Proposition 3}

We have:
\[
\lim_{\phi\rightarrow0^{+}}\mathcal{S}(\phi)=-\infty\text{ and }\text{\ensuremath{\mathcal{S}(1)=\frac{\gamma(2b-1)(\lambda+1)}{\sigma}.}}
\]
Therefore, $S(1)>0$ if $b>\tfrac{1}{2}$ and we conclude that $\mathcal{S}(\phi)$
has at least one zero for $\phi\in\left(0,1\right)$. Further, we
have:
\[
\dfrac{d\mathcal{S}}{d\phi}(\phi)=\dfrac{1}{2\phi^{2}}\left\{ \frac{\gamma(b-1)\left[\lambda\left(\phi^{2}-1\right)+2\phi^{2}\right]}{\sigma}-\frac{2\phi}{\sigma-1}\right\},
\]
whose sign depends on that of the second term, which is a second degree
polynomial and thus has at most two zeros $\left\{ \phi^{-},\phi^{+}\right\} $,
with $\phi^{+}>\phi^{-}$. However, only $\phi^{+}$ lies on the interval
$\phi\in\left(0,1\right)$:
\[
\phi^{+}=\frac{\sigma\left[\frac{1}{\sigma-1}-\sqrt{\frac{\gamma^{2}(b-1)^{2}\lambda(\lambda+2)}{\sigma^{2}}+\frac{1}{(\sigma-1)^{2}}}\right]}{\gamma(b-1)(\lambda+2)}.
\]
Given that the leading coefficient of the polynomial is negative,
we have that $\mathcal{S}(\phi)$ is increasing for $\phi\in\left(0,\phi^{+}\right)$
and decreasing for $\phi\in\left(\phi^{+},1\right)$. Thus, $\mathcal{S}(\phi)$
has at most two zeros for $\phi\in\left(0,1\right)$, called sustain
points $\phi_{s1}$ and $\phi_{s2}$ (with $\phi_{s1}<\phi_{s2}$).\footnote{One of which is given by $\phi=1$ if $b=\frac{1}{2}$.}
If $b<\tfrac{1}{2}$, there exist at most two sustain points $\phi_{s1}\in\left(0,1\right)$
and $\phi_{s2}\in\left(0,1\right)$ and we have $\mathcal{S}(\phi)<0$
for $\phi\in\left\{ \left(0,\phi_{1s}\right)\cup\left(\phi_{2s},1\right)\right\} $
and $\mathcal{S}(\phi)>0$ for $\phi\in(\phi_{1s},\phi_{2s})$. If
$b>\tfrac{1}{2}$, there exists one unique sustain point $\phi_{s1}\in\left(0,1\right)$
and we have $\mathcal{S}(\phi)<0$ for $\phi\in\left(0,\phi_{1s}\right)$
and $\mathcal{S}(\phi)>0$ for $\phi\in\left(\phi_{1s},1\right)$,
which concludes the proof.\hfill{}$\square$

\subsection{Symmetric dispersion}

Using (\ref{eq:stability symmetric dispersion}), the breakpoints are given by:
{\small{}
\begin{align}
\phi_{b1} & =\frac{\sqrt{\gamma^{2}(\sigma-1)^{2}\left[8b(\lambda+1)^{2}-4\lambda^{2}-8\lambda-3\right]+4\gamma\sigma(\sigma-1)+4\sigma^{2}}-2b\gamma(\lambda+1)(\sigma-1)}{\gamma(\sigma-1)\left[2b(\lambda+1)-2\lambda-3\right]-2\sigma}\nonumber \\
\phi_{b2} & =-\frac{\sqrt{\gamma^{2}(\sigma-1)^{2}\left[8b(\lambda+1)^{2}-4\lambda^{2}-8\lambda-3\right]+4\gamma\sigma(\sigma-1)+4\sigma^{2}}+2b\gamma(\lambda+1)(\sigma-1)}{\gamma(\sigma-1)\left[2b(\lambda+1)-2\lambda-3\right]-2\sigma}.\label{eq:break points}
\end{align}
}{\small\par}

\noindent The break point $\phi_{b1}$ lies in the interval $(0,1)$ if and
only if: 

\begin{enumerate}[(i).]
\item  $\gamma\in\left(\dfrac{2 \sigma }{(2 \lambda +1) (\sigma -1)},1\right)$,
\item  $b\in\left[b_{1},b_{2}\right),$
\end{enumerate}
\noindent where:
\begin{align*}
b_{1} & =\frac{\left[\gamma(2\lambda+1)(\sigma-1)-2\sigma\right]\left[\gamma(2\lambda+3)(\sigma-1)+2\sigma\right]}{8\gamma^{2}(\lambda+1)^{2}(\sigma-1)^{2}},\\
b_{2} & =\frac{\gamma(2\lambda+1)(\sigma-1)-2\sigma}{2\gamma(\lambda+1)(\sigma-1)}.
\end{align*}

\noindent We have that $b_{1}\in\left(0,\frac{1}{2}\right)$ and $b_{2}\in\left(b_{1},1\right)$. If, additionally, $b<\frac{1}{2}$, then we have also that $\phi_{b2}\in\left(0,1\right)$.\footnote{If  $\gamma < \frac{2 \sigma }{\lambda  (\sigma -1)}$, then $b_{2}\in\left(b_{1},\frac{1}{2}\right)$ and the condition is trivially met by (ii). In this case, both break points exist.}  This means that the possibility of complete re-dispersion following agglomeration
as $\phi$ increases requires that related variety is neither too
high nor too low. However, if $b>\frac{1}{2}$, then $\phi_{b2}$ does not exist and we have a single break point $\phi_{b1}$ if conditions (i) and (ii) are satisfied.

\subsection{Proof of Proposition 4}

Taking the derivative of $\mathcal{G}$ in (\ref{eq:stabasdisp}) with respect to
$b$ we get:
\[
\frac{\partial\mathcal{G}}{\partial b}=-2\gamma(\sigma-1)(2z-1)^{3}\left(1-\phi^{2}\right)<0,\ \ z\in\left(\frac{1}{2},1\right)
\]
Next, solving $\mathcal{G}$ in (\ref{eq:stabasdisp}) for $b$ yields:{\footnotesize{}
\begin{equation}
b=b_{c}\equiv\frac{(2z-1)\left(1-\phi^{2}\right)\left[\sigma+\gamma(\sigma-1)(1-2z)^{2}\right]-\sigma\left[2z^{2}(\phi-1)^{2}-2z(\phi-1)^{2}+\phi^{2}+1\right]\ln\left[\frac{z(\phi-1)+1}{z(1-\phi)+\phi}\right]}{2\gamma(\sigma-1)(2z-1)^{3}\left(1-\phi^{2}\right)}.\label{eq:bcritconda}
\end{equation}
}{\footnotesize\par}

\noindent As a result, we have $\mathcal{G}>0$ for $b<b_c$ and $\mathcal{G}<0$ for $b>b_c$. Next, we will prove that $b_c<1/2$. 

First, notice that $\lim_{\phi\rightarrow1}b_{c}=\frac{1}{2}.$ Next,
we have:
\[
\dfrac{\partial b_{c}}{\partial\phi}=\frac{\sigma\mathcal{N}}{2\gamma(\sigma-1)(2z-1)^{3}\left(\phi^{2}-1\right)^{2}\left[z(\phi-1)+1\right]\left[z(\phi-1)-\phi\right]},
\]
where:
\begin{align*}
\mathcal{N}= & (2z-1)\left(\phi^{2}-1\right)\left[2z^{2}(\phi-1)^{2}-2z(\phi-1)^{2}+\phi^{2}+1\right]-\\
 & -4\left[z(\phi-1)+1\right]{}^{2}\left[z(1-\phi)+\phi\right]{}^{2}\ln\left[\frac{z(\phi-1)+1}{z(1-\phi)+\phi}\right].
\end{align*}
The numerator of the derivative is negative. As for $\mathcal{N}$,
observe that:
\begin{align*}
\dfrac{\partial\mathcal{N}}{\partial z}= & -4\sigma(2z-1)(\phi-1)^{2}\times\\
 & \times\left\{ (2z-1)\left(1-\phi^{2}\right)+2\left[z(\phi-1)+1\right]\left[z(\phi-1)-\phi\right]\ln\left[\frac{z(\phi-1)+1}{z(1-\phi)+\phi}\right]\right\} .
\end{align*}
The first term inside the curly brackets is positive and the second
one is negative as is the log term. Therefore, we have $\frac{\partial\mathcal{N}}{\partial z}<0$.
Since $\mathcal{N}\left(z=\frac{1}{2}\right)=0$ and given that $\mathcal{N}$
is continuous in $z$, we can conclude that $\mathcal{N}<0$ for $z\in\left(\frac{1}{2},1\right)$.
Thus, we have $\frac{\partial b_{c}}{\partial\phi}>0,$ which means
that $b_{c}<\frac{1}{2}$. Thus, if $b>\frac{1}{2}$, we have $\mathcal{G}<0$ for any value of $\lambda$ such that $\lambda^*(z)>0$. This concludes the proof.

\subsection{Bifurcation at symmetric dispersion}

\noindent We can get a better picture of the dynamic properties of the model
by studying the type of local bifurcation that the symmetric equilibrium
undergoes at some break-point $\phi=\phi_{b}.$ After some tedious calculations,
it is possible to show the following:
\[
\dfrac{\partial f}{\partial z}\left(\dfrac{1}{2};\phi_{b}\right)=0;\ \dfrac{\partial^{2}f}{\partial z^{2}}\left(\dfrac{1}{2};\phi_{b}\right)=0;\ \dfrac{\partial f}{\partial\phi}\left(\dfrac{1}{2};\phi_{b}\right)=0;\ \dfrac{\partial^{2}f}{\partial\phi\partial z}\left(\dfrac{1}{2};\phi_{b}\right)>0,
\]
where $\phi_b\in\{\phi_{b1},\phi_{b2}\}$.
According to \citeauthor{Guckenheimer2002} (2002, pp.~150), the conditions
above ensure that symmetric dispersion undergoes a pitchfork
bifurcation at $\phi=\phi_{b}$. Further, we have{\small{}
\[
\dfrac{\partial^{3}f}{\partial z^{3}}\left(\dfrac{1}{2};\phi_{b}\right)=-\frac{32\mu(1-\phi)}{(\sigma-1)\sigma(\phi+1)^{4}}\xi,
\]
}where
\[
\xi=3\gamma(\sigma-1)\left[b(\phi+1)^{2}-\phi^{2}-1\right]\left[\lambda(\phi-1)+2\phi\right]-\sigma(\phi-1)^{2}(\phi+1).
\]
If $\xi>0$, the derivative is negative, and thus the pitchfork is supercritical and a curve of stable asymmetric equilibria branches from symmetric dispersion to its right. If $\xi<0$, the derivative is positive and hence the pitchfork is subcritical and a curve of unstable asymmetric equilibria branches from symmetric dispersion to its left. If $\xi=0$ we say that the pitchfork is degenerate.

\section{Related variety and economic integration: additional illustrations}

In this Appendix we change the benchmark parameter values to find out whether the value of $b$ can bear more drastic implications to the qualitative structure of the spatial economy. Particularly, we set $(\lambda,\gamma,\sigma)=(1,1,8)$.

\begin{figure}[!h]
     \centering
     \begin{subfigure}{0.4\textwidth}
         \centering
         \includegraphics[width=\textwidth]{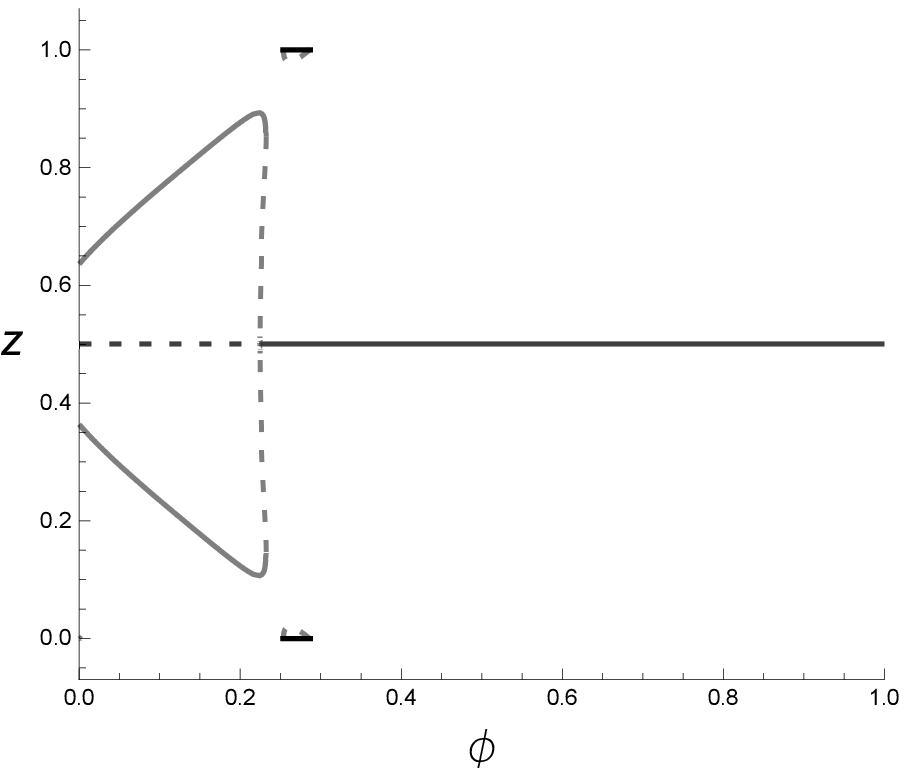}
     \end{subfigure}
     \hfill
     \begin{subfigure}{0.4\textwidth}
         \centering
         \includegraphics[width=\textwidth]{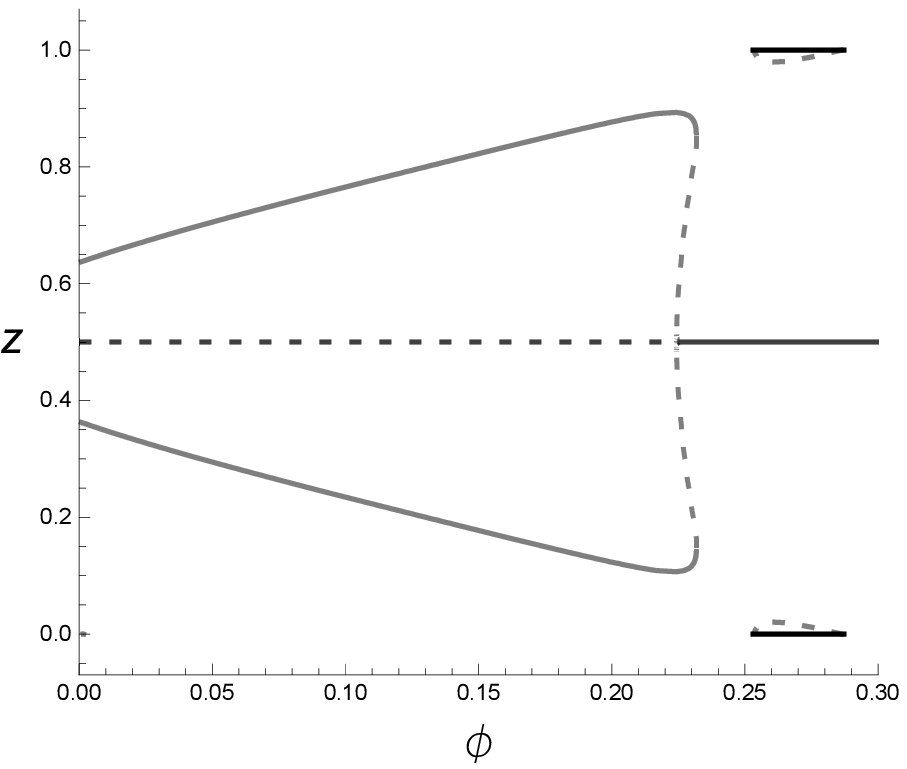}
     \end{subfigure}
     \caption{Bifurcation diagram for $(\lambda,\gamma,\sigma,b)=(1,0.9,8,0.1405)$. Filled lines correspond to stable equilbria and dashed lines correspond to unstable equilibria.}
     \label{bifdiagramscenariostrange1}
\end{figure}

\begin{figure}[!h]
     \centering
     \begin{subfigure}{0.4\textwidth}
         \centering
         \includegraphics[width=\textwidth]{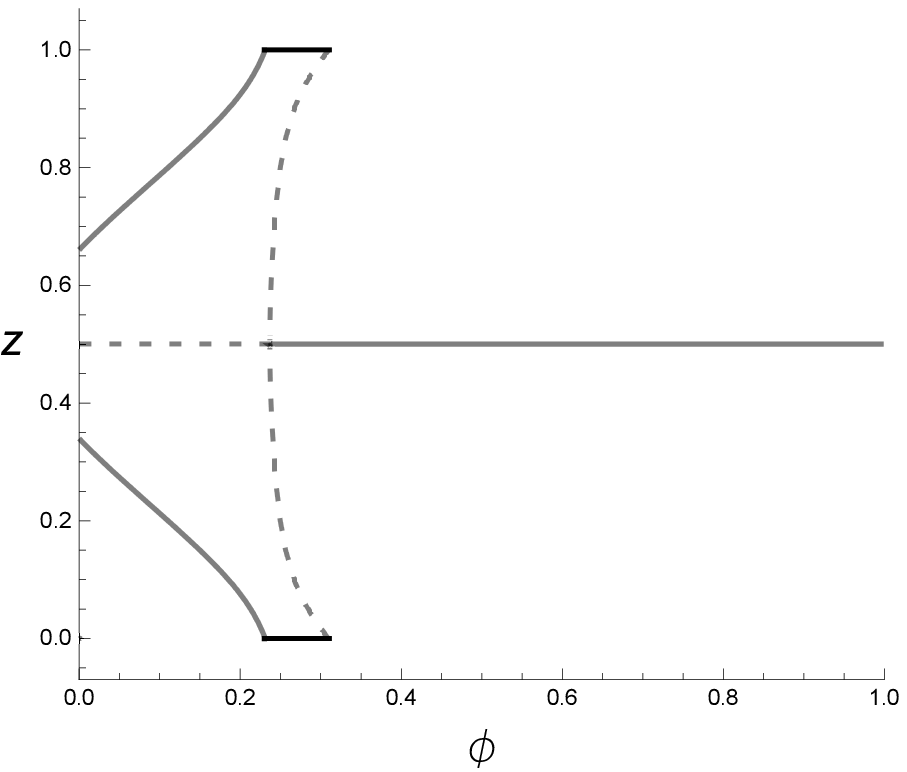}
     \end{subfigure}
     \hfill
     \begin{subfigure}{0.4\textwidth}
         \centering
         \includegraphics[width=\textwidth]{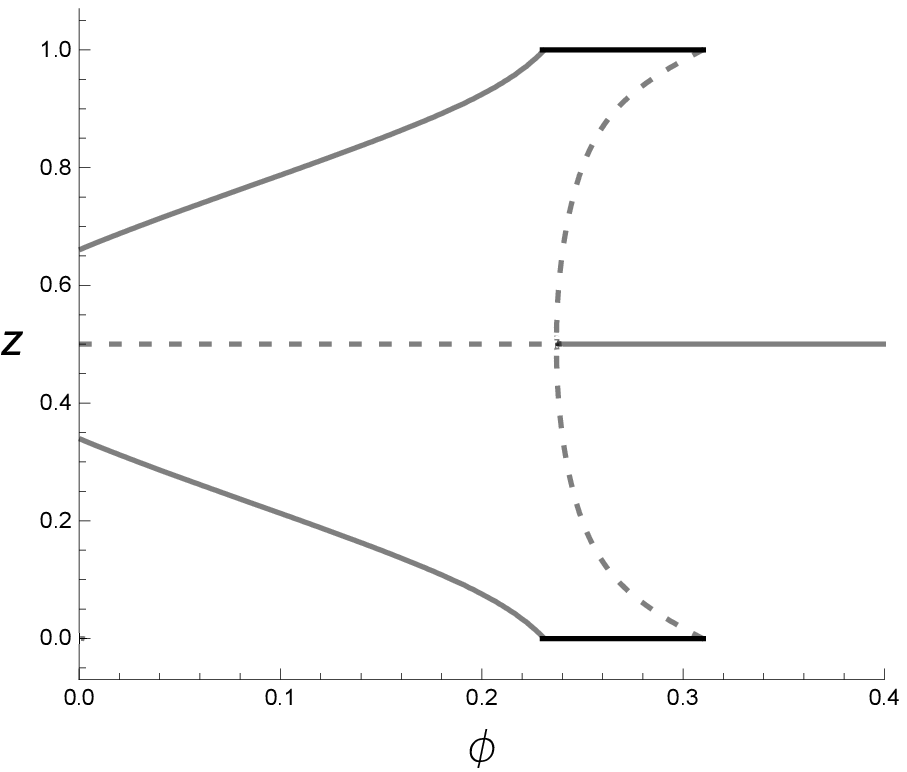}
     \end{subfigure}
     \caption{Bifurcation diagram for $(\lambda,\gamma,\sigma,b)=(1,0.9,8,0.15)$. Filled lines correspond to stable equilbria and dashed lines correspond to unstable equilibria.}
     \label{bifdiagramscenariostrange2}
\end{figure}

In Figure~\ref{bifdiagramscenariostrange1} we set $b=0.146$, which is considerably lower than in the previously illustrated cases. We can observe that the re-dispersion process is quite different. Here, for very low levels of $\phi$ an asymmetric dispersion equilibrium exists and is stable, and becomes more asymmetric as $\phi$ increases. However, after a certain point, the economy starts to re-disperse until $\phi$ finds a limit point $\phi_l\in(\phi_{b2},1)$ above which the economy suddenly re-disperses evenly among the two regions. That is, the asymmetric dispersion undergoes a saddle-node bifurcation (refer to the end of Section 4 for a more detailed discussion) and there exists locational hysteresis as both an asymmetric dispersion equilibrium and the symmetric dispersion equilibrium are simultaneously stable for $\phi\in(\phi_{b2},\phi_l)$. The symmetric equilibrium undergoes a subcritical pitchfork bifurcation at $\phi=\phi_{b2}$.

What is perhaps more striking though is that, as $\phi$ increases further, it encounters an interval $\phi \in (\phi_{s1},\phi_{s2})$ whereby both agglomeration and symmetric dispersion are stable, and a curve of unstable asymmetric equilibria exists in between. The striking feature is that the curve of agglomeration equilibria is not connected to any other kind of equilibrium. However, further increases in $b$ will eventually connect the asymmetric equilibrium curve with the full agglomeration as shown in Figure~\ref{bifdiagramscenariostrange2}. This apparently strange behaviour may be attributed to the limitations imposed by the implicitly adopted \emph{asymptotic} stability as the dynamic stability criterion in this paper. One way to investigate this issue could be to employ \emph{strategic} stability as in \cite{Demichelis2003}, which entails additional necessary conditions for stability which restrict the equilibrium set.\footnote{We thank Anna Rubinchik for this reference on  stability conditions.} The potential game's approach taken recently by \cite{OsawaAkamatsu2020} could also be analysed in this context for equilibrium refinement. Finally, the connection via an increase in $b$ could hint at the existence of co-dimension $2$ bifurcations with both $\phi$ and $b$ employed as bifurcation parameters.\footnote{We thank Sofia B.S.D. Castro for pointing out this potentially relevant issue.} However, it can be shown that choosing $b$ as the \emph{alternative } bifurcation parameter leads to the same  local bifurcations   as $\phi$. Although these are all relevant points, we do not pursue these issues further in our paper.

In Figure~\ref{bifdiagramscenariostrange2}, we have $b=0.15$, and the economy reaches full agglomeration smoothly as $\phi$ increases but then jumps discontinuously to symmetric dispersion. We have the single breakpoint $\phi_{b2}\in (\phi_{s1},\phi_{s2})$, which means that there exists locational hysteresis as for $\phi \in (\phi_{b2}, \phi_{s2})$ both agglomeration and symmetric dispersion are simultaneously stable.

Further increases in $b$ can be shown to lead, first to a situation similar to that of Figure~\ref{bifdiagramscenariov}, and then to ubiquitous agglomeration for any level of $\phi$.

\section{Multiplicative case: a numerical analysis}

We now consider that the interaction between scientists hailing from different regions is of multiplicative nature. For instance, a Cobb-Douglas specification such as:
\begin{equation}
\Phi_{i}^{CD}(s)=\min\left\{ \dfrac{z_{i}^{b}z_{j}^{1-b}}{a_{i}(s)}\gamma A,1\right\},\label{eq:prob of succesful innovation-3}
\end{equation}
yields a multiplicative scenario, and thus brings our setup closer to \cite{Berliant-Fujita-IER2008}. Again we shall impose a parametrization that guarantees that the first term lies in the interval $[0,1]$.

Regarding agglomeration, it cannot be an equilibrium because $\Delta v(1) = \frac{\mu\ln\phi}{\sigma-1}<0$. Since no innovation occurs, expected profits are driven down to zero and the only thing that matters is the cost-of-living, which is positive in the fully agglomerated region. However, if  an agent moves to the ``empty'' region, innovation occurs through inter-regional interaction and he will earn a positive nominal wage, which means that exogenous perturbations always increase the utility of agents. While this may strike as an implausible outcome, one way to counter this would be to specify  a different probability for corner equilibria, potentially following \cite{Berliant-Fujita-IER2008}. 

The \textit{symmetric dispersion} retains the qualitative properties of the benchmark case analyzed in Section 3.2.2, since $\Phi_i^{CD}(s)$ in \eqref{eq:prob of succesful innovation-3} is a particular case of  $\Phi_i^G(s)$ in \eqref{eq:general}. That is, there may  exist two break points $\phi_{b1}$ and $\phi_{b2}$ given by (\ref{eq:break points}) (provided the conditions presented therein hold) and symmetric dispersion is stable for $\phi \in (0,\phi_{b1})\cup (\phi_{b2},1)$,
and unstable for $\phi\in\left(\phi_{b1},\phi_{b2}\right)$.

We can grasp the general qualitative behaviour and properties of the model under $\Phi_{i}^{CD}(s)$ by depicting a gallery of bifurcation diagrams for several different values of $b$. The benchmark parameter values are $(\sigma,\gamma,\lambda)=(8,1,4)$, while the value for $b$ is reported in the caption of each picture.

\begin{figure}
     \centering
     \begin{subfigure}[b]{0.35\linewidth}
         \centering
         \includegraphics[width=\linewidth]{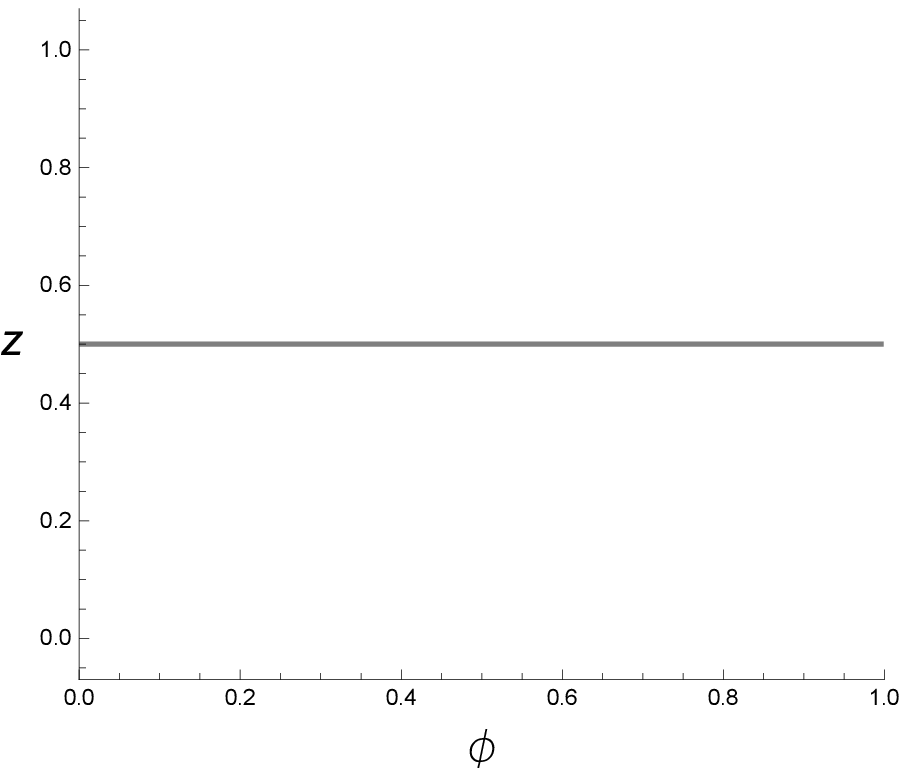}
         \caption{Stable symmetric dispersion: $b=0.1$.}
         \label{fig:bifa}
     \end{subfigure}
     \hfill
     \begin{subfigure}[b]{0.35\linewidth}
         \centering
         \includegraphics[width=\linewidth]{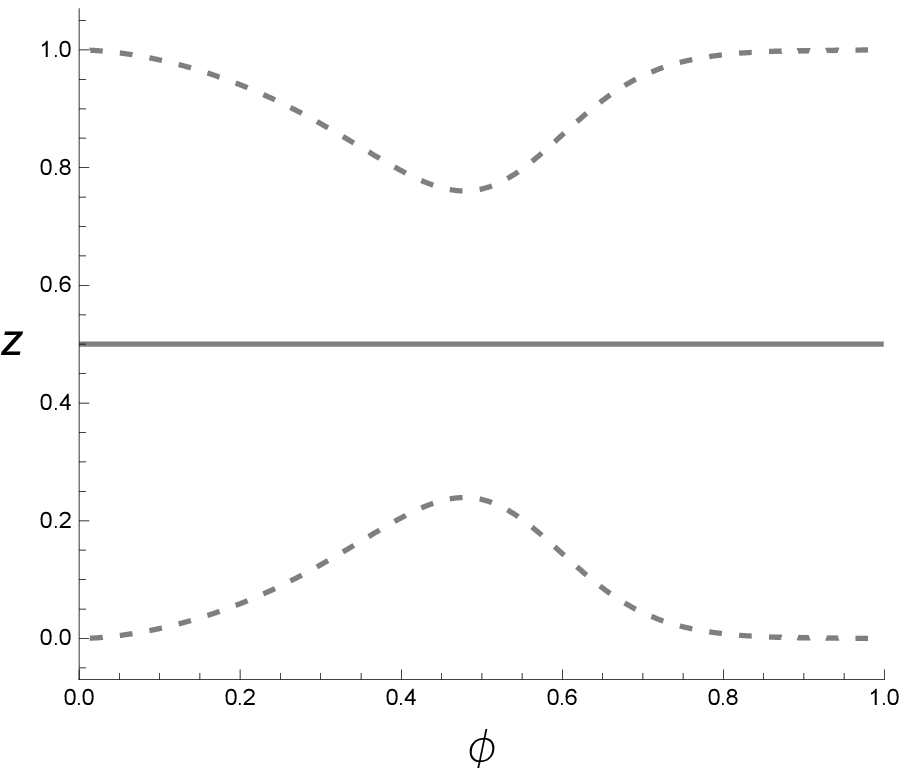}
         \caption{Stable symmetric dispersion: $b=0.44$.}
         \label{fig:bifb}        
     \end{subfigure}
     \par\bigskip
          \begin{subfigure}[b]{0.35\linewidth}
         \centering
         \includegraphics[width=\linewidth]{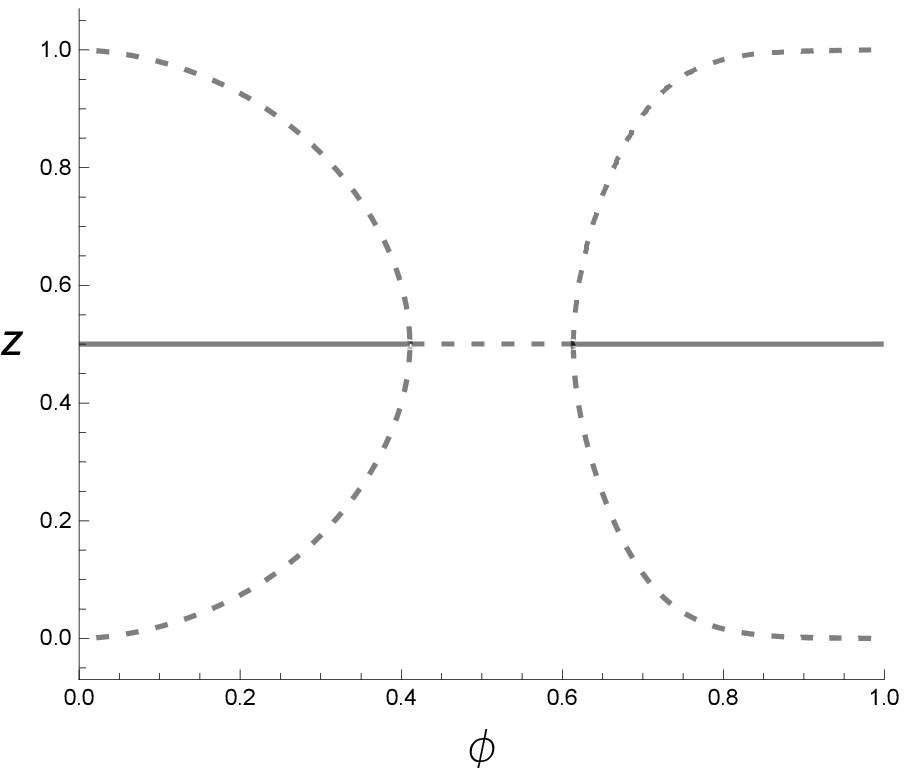}
         \caption{Stable symmetric dispersion for $\phi<\phi_{b1}$ and $\phi>\phi_{b2}$: $b=0.45$.}
         \label{fig:bifc}
     \end{subfigure}
     \hfill
\begin{subfigure}[b]{0.35\linewidth}
         \centering
         \includegraphics[width=\linewidth]{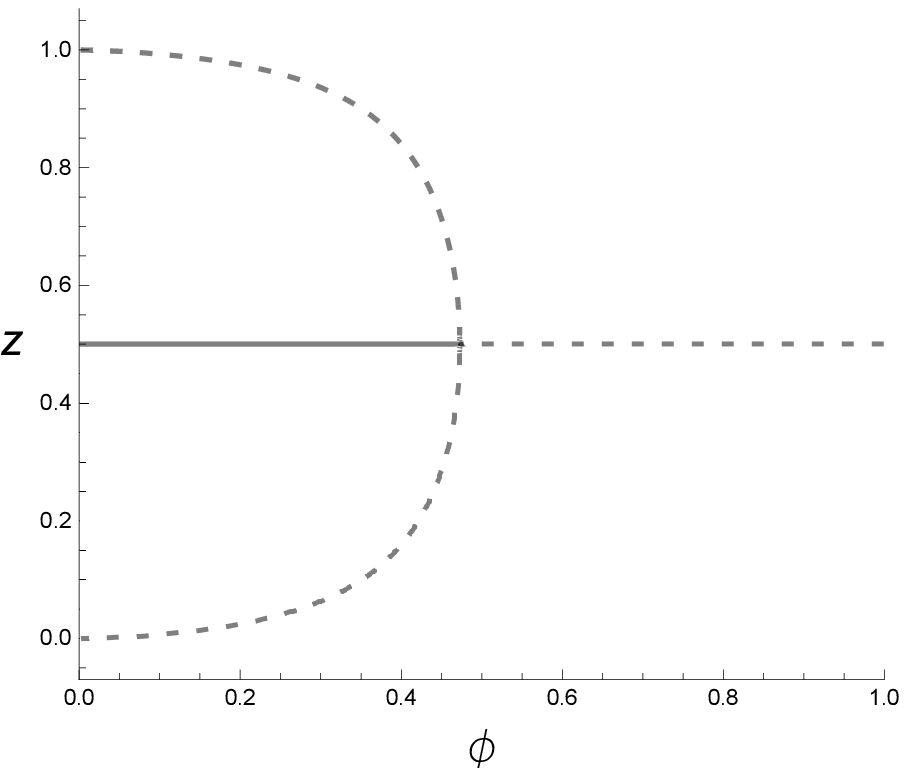}
         \caption{Stable symmetric dispersion for $\phi<\phi_{b1}$: $b=0.5$.}
         \label{fig:bifd}        
     \end{subfigure}
     \par\bigskip
     \begin{subfigure}[b]{0.35\linewidth}
         \centering
         \includegraphics[width=\linewidth]{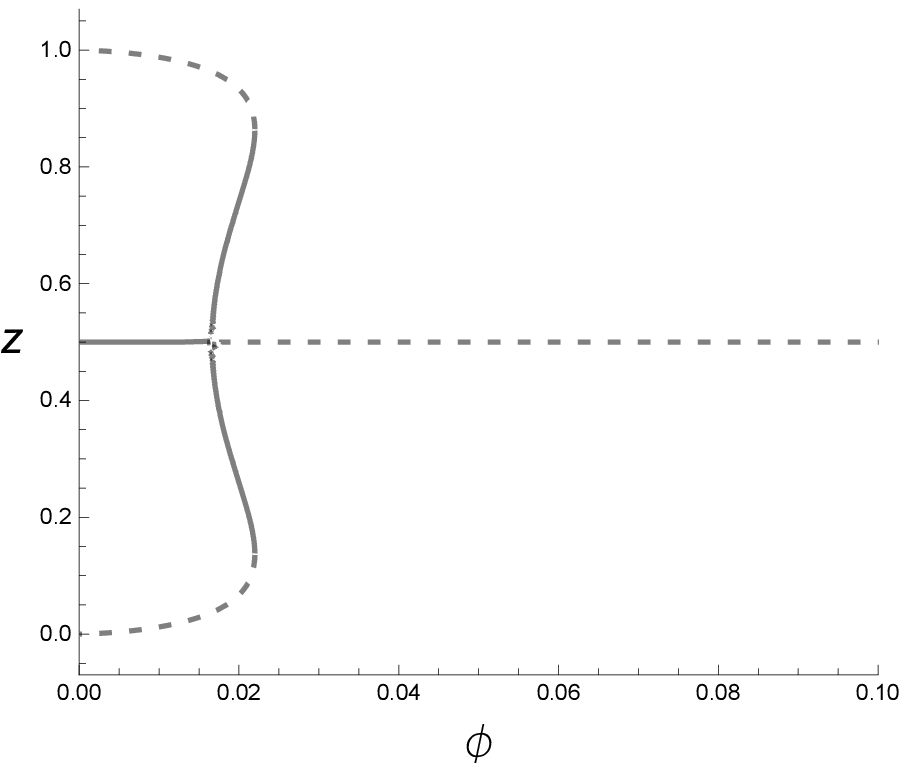}
         \caption{Stable symmetric dispersion for $\phi<\phi_{b1}$ and stable asymmetric dispersion for $\phi>\phi_{b1}$: $b=0.65$.}
         \label{fig:bife}        
     \end{subfigure}
         \hfill
\begin{subfigure}[b]{0.35\linewidth}
         \centering
         \includegraphics[width=\linewidth]{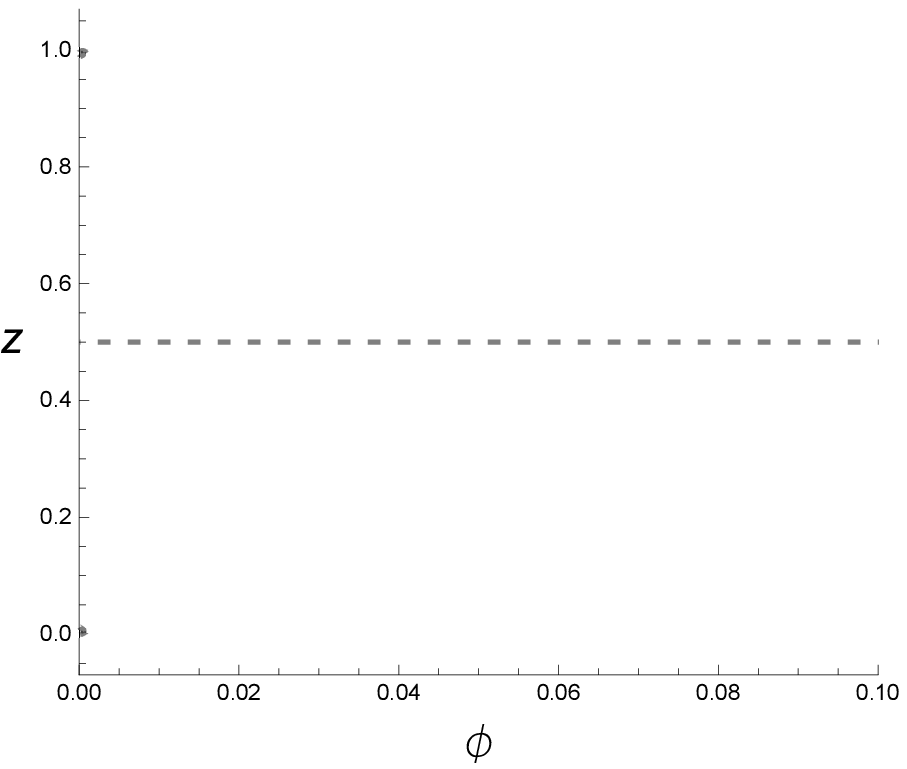}
         {\caption{\\ No stable equilibria: $b=0.75$}.}
         \label{fig:biff}        
     \end{subfigure}
\caption{ Bifurcation diagram for $(\lambda,\gamma,\sigma)=(4,1,8)$. Filled lines correspond to stable equilibria and dashed lines correspond to unstable equilibria.}
\label{fig:Figure12}
\end{figure}
As $b$ increases from a very small value, initially only symmetric dispersion exists and is stable for the entire range of $\phi$ (Figure~12(a)). As $b$ increases, there emerge two curves of unstable asymmetric dispersion equilibria (Figure~12(b)), one with a minimum and the other with a maximum (for the same $\phi$, due to symmetry). For higher levels of $b$, the  extrema collide vertically at symmetric dispersion and two curves of unstable asymmetric dispersion equilibria branch from two break points. Symmetric dispersion is stable below the lowest break point and above the highest break point, and unstable in between (Figure~12(c)). Once related variety becomes high enough ($b\geq 1/2$), re-dispersion is no longer possible, as expected, and we end up with a subcritical pitchfork bifurcation (Figure~12(d)), as in \cite{Fujita-Krugman-Venables-Book1999}, whereby symmetric dispersion is stable for low values of $\phi$ and becomes unstable for high values of $\phi$. The main difference is that there are no stable equilibria above the break point. If $b$ increases even more, we end up with a supercritical pitchfork bifurcation whereby a stable symmetric dispersion loses stability for $\phi$ large enough (Figure~12(e)). This state encounters a primary branch of stable asymmetric equilibria that, apparently, undergoes a secondary saddle-node bifurcation.\footnote{This qualitative scenario is very similar to the one encountered by \citeauthor{Castro_2021} (2021, Fig. 4 in pp.~197) except that the stability of equilibria is ``reversed'' and the bifurcation parameter employed by them is a heterogeneity parameter.} Finally, for a prohibitively high value of $b$, there are no stable equilibria in the model, as shown by Figure~12(f).


\bigskip{}

\clearpage\singlespacing \bibliographystyle{apalike}
\bibliography{refs}

\end{document}